\documentclass[a4paper,11pt]{article}
\usepackage{amssymb}
\usepackage[dvips]{color}
\usepackage{xcolor}
\usepackage{graphicx}
\usepackage{amsmath,enumerate}
\usepackage[T1]{fontenc}
\usepackage{a4wide}
\usepackage{natbib}
\usepackage{titling}

\newtheorem{theorem}{Theorem}

\newtheorem{assumption}[theorem]{Assumption}

\newtheorem{definition}[theorem]{Definition}
\newtheorem{example}[theorem]{Example}

\newtheorem{lemma}[theorem]{Lemma}

\newtheorem{proposition}[theorem]{Proposition}
\newtheorem{remark}[theorem]{Remark}

\newenvironment{proof}[1][Proof]{\textbf{#1.} }{\ \rule{0.5em}{0.5em}}

\newcommand{\R} {\ensuremath{\mathbb{R}}}
\newcommand{\N} {\ensuremath{\mathbb{N}}}

\newcommand{\Acal}{{\mathcal A}}

\newcommand{\Ccal}{{\mathcal C}}

\newcommand{\Fcal}{{\mathcal F}}

\newcommand{\Hcal}{{\mathcal H}}

\newcommand{\Mcal}{{\mathcal M}}

\newcommand{\Vcal}{{\mathcal V}}

\newcommand{\cA}{{\mathcal A}}
\newcommand{\cB}{{\mathcal B}}
\newcommand{\cC}{{\mathcal C}}

\newcommand{\cF}{{\mathcal F}}

\newcommand{\cH}{{\mathcal H}}
\newcommand{\cI}{{\mathcal I}}

\newcommand{\cK}{{\mathcal K}}

\newcommand{\cM}{{\mathcal M}}
\newcommand{\cN}{{\mathcal N}}
\newcommand{\cP}{{\mathcal P}}
\newcommand{\cQ}{{\mathcal Q}}
\newcommand{\cR}{{\mathcal R}}

\newcommand{\cX}{{\mathcal X}}
\newcommand{\cY}{{\mathcal Y}}

\newcommand{\normc}[1]{\|#1\|_{\infty}}					
\newcommand{\normW}[1]{\|#1\|}							
\DeclareMathOperator{\cl}{cl_{\infty}}					
\DeclareMathOperator{\clw}{cl_{\sigma}}					

\DeclareMathOperator{\lin}{lin}
\DeclareMathOperator{\graph}{graph}

\newcommand{\Lzero}{\mathbb{L}^0}						
\newcommand{\Lzeroplus}{\mathbb{L}^0_+}					
\newcommand{\Linf}{\mathbb{L}^\infty}					
\newcommand{\Linfplus}{\mathbb{L}^\infty_+}				

\newcommand{\LW}{\mathbb{L}}							
\newcommand{\LWplus}{\LW_+}							

\newcommand{\Lzerop}{L^0_P}							
\newcommand{\Linfp}{L^\infty_P}						
\newcommand{\Linfplusp}{(L^\infty_P)_+}					
\newcommand{\Lonep}{L^1_P}								

\newcommand{\polar}{\cN}								
\newcommand{\KM}{\cK_\lambda}							
\newcommand{\HM}{\cH_\lambda}							
\newcommand{\CM}{\cC_\lambda}							
\newcommand{\tildeC}[1]{\widetilde{\cC}_{#1}}			
\newcommand{\CMK}[1]{\cC_{\lambda,#1}}					
\newcommand{\Qapp}{\mathfrak{Q}_{app}}					
\newcommand{\mtg}{\mathfrak{Q}}							
\newcommand{\QN}{Q_n}
\newcommand{\PW}{P^W}									
\newcommand{\cPW}{\cP^W}								
\newcommand{\LPW}{\lin(\cPW)}								
\newcommand{\jp}{j_P}									

\def\qs{\mathrm{q.s.}}  			
\def\NA{\mathrm{(NA)}}				
\def\NFLVR{\mathrm{(NFLVR)}}		

\def\sNA{\mathrm{(sNA)}}				
\def\sNFLVR{\mathrm{(sNFLVR)}}		
\def\sNABR{\mathrm{(sNABR)}}		

\thanksmarkseries{arabic}

\definecolor{Verde}{rgb}{0,0.7,0}

\begin{document}

\title{Arbitrage-free modeling under Knightian Uncertainty}
\author{M. Burzoni\thanks{Mathematical Institute, University of Oxford, Andrew Wiles Building, Woodstock Road, Oxford OX2 6GG, email: matteo.burzoni@maths.ox.ac.uk. Part of the research work has been conducted at ETH Z\"{u}rich, the ETH foundation is gratefully acknowledged. Part of the research work is supported by the Hooke Research Fellowship from the University of Oxford.}\and M. Maggis\thanks{Department of Mathematics, University of Milan, Via Saldini 50, 20133 Milan, email: marco.maggis@unimi.it}}
\maketitle
 
 \begin{abstract}
 	 We study the Fundamental Theorem of Asset Pricing for a general financial market under Knightian Uncertainty. We adopt a functional analytic approach which requires neither specific assumptions on the class of priors $\cP$ nor on the structure of the state space. Several aspects of modeling under Knightian Uncertainty are considered and analyzed. 
 	 We show the need for a suitable adaptation of the notion of No Free Lunch with Vanishing Risk and discuss its relation to the choice of an appropriate technical filtration.
	In an abstract setup, we show that absence of arbitrage is equivalent to the existence of \emph{approximate} martingale measures sharing the same polar set of $\cP$.
	We then specialize our results to a discrete-time financial market in order to obtain martingale measures.
 \end{abstract}



 
 \section{Introduction}
 
The mathematical modeling of financial markets is a challenging task initiated over a century ago by Bachelier \cite{Bachelier}, who firstly observed how the oscillations of the prices on stock exchanges could be represented as the trajectories of the Brownian Motion. After the major contributions by \cite{BlackScholes:73} and \cite{Merton:73}, an outbreak of sophisticated mathematical models for Finance was observed in the last decades in the scientific literature.
For any of such various models the \emph{absence of arbitrage} is a foundational principle. 
According to this condition, it is not possible to make a positive gain without taking any shortfall risk.
This is not only a reasonable feature of the model but also a property which is typically satisfied by real markets. 
Indeed, it is widely accepted that markets are \emph{efficient}: even if an arbitrage opportunity occurred, it would soon vanish as the traders willing to exploit it would cause a change in the underlying prices.
A cornerstone result is the so-called Fundamental Theorem of Asset Pricing (FTAP) which establishes the equivalence between absence of arbitrage and the existence of suitable pricing functionals.
A general FTAP was firstly proved by \cite{DMW90} in a discrete time setting under No Arbitrage (NA) and by \cite{DS94} in continuous-time under the stronger requirement No Free Lunch with Vanishing Risk (NFLVR). 
A key aspect of such breakthrough results is the closure property of the cone of super-replicable claims at zero cost, which can be derived from the no-arbitrage condition.

Mathematically, the classical literature flourished on the standing assumption that a probability space $(\Omega,\cF,P)$ is given. The role of $P$ is essentially to establish the class of events which are irrelevant for the model, the $P$-nullsets. Such an assumption is very much exposed to the so-called \emph{model risk}, namely, the fact that the outputs produced by the model are sensitive to the choice of $P$ and a wrong choice might lead to severe consequences.
It is therefore natural to wonder whether a FTAP can be established in a more robust setting by the mean of a suitable functional analytic approach.

The \emph{quasi-sure setting} consists in replacing the given probability measure $P$ with a class of probability measures $\cP$ aiming at capturing the model ambiguity that agents are facing. In this setting only events which are irrelevant with respect to \emph{any} of the considered priors are deemed impossible.
In the seminal paper
\cite{BN13}, the authors construct a discrete time framework, inspired by dynamic programming ideas, in order to prove a quasi-sure version of the FTAP. Such a framework has become standard in the Robust Finance literature and many results can be obtained within the same setting (see, e.g., \cite{BayraktarZhangMOR,BZ17,BC19} in discrete-time and \cite{BBKN14} in continuous-time).
An alternative \emph{pathwise setting} has been considered in the literature (see, e.g., \cite{AB13,BKN17,BKPT19,Pointwise,BFM16,DS13,HO15,OW18,Riedel}), where instead of a probabilistic formulation of the problem the authors work directly on the set of scenarios $\Omega$.

The first main result of the paper (Theorem \ref{FTAP}) is an abstract version of the FTAP in a general quasi-sure setting and it is the content of Section \ref{sec:main}.
The novelty is that its proof relies only on functional analytic arguments, therefore, we do not require a structure which is amenable to the use of measurable selection techniques, as customary in the aforementioned literature. 
As in the classical case, where a reference measure is given, the advantage is that both discrete and continuous time models can be attacked in the same way since the problem reduces to show the weak-closure of the cone of superreplicable claims.
Moreover, these techniques are typically applied in a multitude of other related problems (among others the super-hedging duality and utility maximization).
We consider a notion of arbitrage that we call $\cP$-\emph{sensitive} (see Definition \ref{arbitrage:sensitive} and the subsequent discussion).
The peculiarity of this condition is that every market participant agrees on the presence of such arbitrage opportunities but they may well disagree on which strategy will realize it. These models cannot represent equilibrium prices for the underlying securities and they should be excluded. The sensitive versions of arbitrage are equivalent to their quasi-sure versions in the frameworks of \cite{BN13} and \cite{STZ11a}. In Section \ref{sec:C} we explain how the role of the chosen filtration is crucial in order to recover the standard notions, in relation to the problem of aggregation of $P$-dependent strategies. This is a technical aspect of the chosen mathematical framework and, from an economic point of view, whether the considered market is viable or not should not depend on that. To this extent the sensitive notions of arbitrage captures market inefficiencies beyond the necessary mathematical technicalities.   

Starting from the general FTAP we show how stronger results can be obtained by allowing for more structure. In particular, several aspects of modeling under Knightian Uncertainty are discussed throughout the paper. 
In Section \ref{discrete:setting} we analyze the discrete-time setting and show the convergence of the \emph{approximate} martingale measures of Theorem \ref{FTAP} to \emph{true} martingale measures. This is proven in Theorem \ref{FTAP_compact} and we illustrate one application in the context of
Martingale Optimal Transport (MOT). 
Finally, Section \ref{sec:proofs} contains all the proofs of the main results.

We conclude this Introduction with the frequently used notation.


\subsection{Notations and setup}\label{notations}
Let $\Omega$ be a separable metric space and $\cF$ the associated sigma algebra of Borel measurable events.
We let $\cM_1$ be the class of probability measures on $(\Omega,\cF)$ endowed with the usual weak topology $\sigma(\cM_1,\cC_b)$, where $\cC_b$ is the space of continuous and bounded functions on $\Omega$. Given $\cP_1,\cP_2\subset \cM_1$ we define 
\begin{description}
	\item[$\cP_1\ll \cP_2$] if $\sup_{P\in\cP_2}P(A)=0$ implies $\sup_{P\in\cP_1}P(A)=0$. We say that $\cP_2$ dominates $\cP_1$;
	\item[$\cP_1\approx \cP_2$] if both $\cP_1\ll \cP_2$ and $\cP_2\ll \cP_1$ holds. We say that $\cP_1$ and $\cP_2$ are equivalent.
	\item[$P \lll \cP$] if there exists a $P_1\in\cP$ such that $P\ll P_1$.
\end{description}
For a given $\cP\subset\cM_1$, we introduce the vector space of countably additive signed measures dominated by $\cP$, namely $ca(\cP)$ (see \cite{MMS18}).
We shall denote by $\polar$ the family of \emph{polar sets}, namely, $$\polar:=\{A\subset A' \mid A'\in\cF \text{ and } P(A')=0\  \forall P\in\cP\}.$$
A statement is said to hold quasi surely ($\qs$) if it holds outside a polar set. 
It is possible to identify measurable functions which are $\qs$ equal and $\Lzero$ will indicate the quotient space. 
$\Linf$ is the subspace of $\qs$ bounded functions, which we endow with the norm
$$
\normc{X}:=\inf\left\{m\in\R\mid P(\{|X|>m\})=0\  \forall P\in\cP\right\}.
$$
If no confusion arises we will denote the $\qs$ partial ordering by $\leq$ (resp. $\geq$ and $=$), meaning that for any $X,Y\in \Lzero$, $X\leq Y$ if and only if $P(\{X> Y\})=0$  for every $P\in\cP$.  $(\Linf, \normc{\cdot})$ endowed with the $\qs$ order $\leq$ is a Banach lattice.
Throughout the text we will be given a positive random variable $W\ge 1$, and work with the space 
$$\LW:=\{X\in\Lzero\mid X/W\in\Linf\},$$ 
paired with the norm $\normW{X}:=\normc{X/W}$.  
We finally introduce $\LWplus$, $\Lzeroplus$ and $\Linfplus$, as the subsets of $\qs$ non-negative functions in $\LW$, $\Lzero$ and $\Linf$ respectively. Given a set $\cA\subset \LW$, $\cl{\cA}$ will denote the closure with respect to $\normc{\cdot}$ of $\cA\cap \Linf$.


\section{An abstract formulation of the FTAP}\label{sec:main}
Fix a measurable space  $(\Omega,\cF)$ and $W\in\Lzeroplus$ with $W\ge 1$. 
The financial market is described, in an abstract form, by the set of financial contracts attainable at zero cost denoted by $\cK\subset \Lzero$. Throughout the text $\cK$ is assumed to be a convex cone.
\begin{definition}\label{def:arbitrage} 
	Let $\cK\subset \Lzero$ be a convex cone.
	\begin{itemize} 
		\item $k\in \cK$ is an \emph{arbitrage opportunity} if $k\in\Lzeroplus\setminus\{0\}$;
		\item $\xi\in \Lzeroplus\setminus\{0\}$ is a \emph{free lunch with vanishing risk} if there exist $c_n\downarrow 0$ and $\{k_n\}\subset\cK$ such that $c_n+k_n\geq \xi$; 
	\end{itemize}
	We denote by $\NA$ and $\NFLVR$ absence of arbitrage and free lunch with vanishing risk respectively. 
\end{definition}
We let $\KM:=\cK\cap \{X\in \Lzero\mid X\geq -\lambda W\}$ for $\lambda\ge 0$ and define
\begin{eqnarray}
	\cC&:=&\left\{ X\in \LW\mid X\leq k \text{ for some } k\in
	\cK\right\},\label{cone:abstract}\\
	\CM&:=&\left\{ X\in \LW\mid X\leq k \text{ for some } k\in
	\KM\right\},\label{cone:abstract:M}
\end{eqnarray}
where we recall that all inequalities are meant to hold $\qs$.

\begin{remark} In the classical dominated case ($\cP\ll P$ for some $P\in\cM_1$), $\cK$ is the class of stochastic integrals of admissible strategies. The use of a random lower bound in the admissibility condition is not new and was used for instance in \cite{BiaginiFrittelli}. An alternative possible choice for $W$ is $W=1$ for which $\KM$ is the set of contracts bounded from below by $-\lambda$, a typical constraint for continuous time models which excludes doubling strategies. In Section \ref{discrete:setting} we show that, in discrete time, a suitable choice for $W$ identifies the class of admissible bounded strategies.
	Under uncertainty, the stochastic integral can be defined in the same way for the discrete time case, since it amounts to a finite sum; in continuous time, it requires a different construction (see e.g.\ \cite{DS13,PP16,STZ11a,Vovk12}).
\end{remark}

Both $\cC$ and $\CM$ are convex and monotone\footnote{We say that that $\cA\subset \Lzero$ is monotone if $Y\le X$ and $X\in \cA$ implies $Y\in\cA$.} sets containing $0$, in addition $\cC$ is also a cone. 
They represent the class of claims which can be super-replicated at zero initial cost by mean of attainable payoffs in $\cK$ and $\KM$ respectively.

As in the classical literature, we can reformulate the no-arbitrage conditions in terms of the cone $\cC$, i.e.,
\begin{eqnarray*}
	\text{$\NA$} &\Longleftrightarrow& \Ccal\cap
	\Linfplus=\{0\}\\
	\text{$\NFLVR$} &\Longleftrightarrow& \cl(\Ccal)\cap
	\Linfplus=\{0\},
\end{eqnarray*}
where $\cl$ denotes the closure with respect to $\normc{\cdot}$ of $\cC\cap \Linf$.
In the context of Knightian Uncertainty this straightforward generalization of the classical concepts might not be sufficient for deriving a general no-arbitrage theory.

\paragraph{Sensitivity: from dominated to non-dominated frameworks.} 
The notion of sensitivity was introduced in \cite{MMS18} and, as we discuss below, it should be interpreted in terms of aggregation of trading strategies with respect to the different measures in the set $\cP$. 
For $P\ll\cP$, we define the linear (projection) map
\begin{equation}\label{projection:P}\begin{array}{rccc}
		j_P:&\Lzero& \to& \Lzerop\\
		& X&\mapsto&[X]_P
	\end{array}
\end{equation}
where $[X]_P$ is the $P$-equivalence class of $X$ in $\Lzerop$. 
\begin{definition}\label{sensitive} A set $\Acal\subset \Lzero$ is called \emph{sensitive} if there exists a family $\cR\subset \Mcal_1$ with $\cR\ll \cP$ such
	that \begin{equation*} \Acal = \bigcap_{P\in\cR}
		j^{-1}_P(j_P(\Acal)).
	\end{equation*}  
	The set $\cR$ will be called
	\emph{reduction set} for $\Acal$.
\end{definition}
We will typically use $\cP$ itself as a reduction set.
As the space $\Lzero$ does only depend on the polar sets $\polar$ we could alternative choose any $\cP'\approx \cP$ and we will occasionally do so.

To better understand the previous definition let us consider for a moment the \textbf{dominated} setting, namely, suppose there exists a reference probability $P$ equivalent to the family $\cP$\footnote{If $\cP \ll P'$ for some $P'\in\cM_1$, the Halmos Savage Lemma (see \cite{HS49}) implies that there exists a probability $P$ which is equivalent to $\cP$.}. In this case $\cC$ is composed of equivalence classes with respect to $P$ and contains, in particular, all the indicators of $P$-null sets. 
It is well know that, in a discrete framework, $\NA$ is equivalent to the existence of martingale measures for $S$ or, stated otherwise,
\begin{equation}\label{FTAP_P}
	\cC\cap
	\Linfplusp=\{0\}
	\Longleftrightarrow
	\cC^0_1\approx P,
\end{equation}
where the set $\cC^0_1:=\{Q\in \cM_1\mid E_Q[X]\leq
0 \;\forall\,X\in \cC\}$ is necessarily composed of measures which are absolutely continuous with respect to $P$.
The no-arbitrage condition \eqref{FTAP_P} could be trivially rewritten as $j_P(\cC)\cap
\Linfplusp=\{0\}$ since, in the dominated case, the map $j_P$ is obviously the identity.
Let us now take into account a class of \textbf{non-dominated} probabilities $\cP\subset\cM_1$.
In order to embed $j_P(\cC)$ in $\LW$, for any $P\in\cP$, we need to consider all the elements in $\LW$ which coincide $P$-a.s.\ with an element of $j_P(\cC)$. 
More precisely, we introduce the set $j_P^{-1}(j_P(\cC))$ and the no-arbitrage condition with respect to a single $P$ would read as $j_P^{-1}(j_P(\cC))\cap\Linfplus=\{0\}$. 
Since we have a whole class of non-dominated $\cP$, representing the uncertainty of the model, we need to consider the set 
\begin{equation}\label{def:Ctilde}
	\tildeC{}=\bigcap_{P\in\cP} j^{-1}_P(j_P(\Ccal)).
\end{equation}
We stress that any $j^{-1}_P(j_P(\Ccal))$ is a subset of $\LW$, i.e., it is a collection of \emph{quasi-sure} equivalence classes of contingent claims. 
To understand how the set $\tildeC{}$ is related to the probabilistic models in $\cP$ we can write, more explicitly, 
\begin{eqnarray*}
	\tildeC{} 
	& = &  \{X\in\LW\mid \forall P\in\cP, \exists k^P\in\cK\text{ s.t. }  X\leq k^P\ P\text{-a.s.}\}
\end{eqnarray*}
The set $\tildeC{}$ induces the following no-arbitrage conditions.
\begin{definition}\label{arbitrage:sensitive} 
	We say that it holds:
	\begin{eqnarray*}
		\text{$\sNA$} &\dot{\Longleftrightarrow}& \tildeC{}\cap
		\Linfplus=\{0\}\\
		\text{$\sNFLVR$} &\dot{\Longleftrightarrow}& \cl(\tildeC{})\cap
		\Linfplus=\{0\},
	\end{eqnarray*}
	where we have emphasized in the acronyms that these are the sensitive versions of the previous notions.
\end{definition}
When $\sNA$ is violated there exists a non-negative contingent claim for which it is not possible to assign a reasonable price. On the one hand, $X$ is non-negative \emph{quasi-surely}, thus, its price cannot be less or equal than zero as it would create an obvious arbitrage. On the other hand a positive price for $X$ disagrees with the output of \emph{any plausible model}. Indeed, since $X\in\tildeC{}$,  $X$ can be super-hedged at zero cost under any of the priors and it will \emph{not} be traded at a positive price. The \emph{same} claim $X$ induces arbitrage opportunities in every model $P$. This means that the inefficiency of the market is identified by every market participant, nevertheless, they might well disagree which strategy should be implemented in order to exploit an arbitrage\footnote{This situation is reminiscent of the example of the two call options with different strikes but same price given in \cite{DH07}.}. Arguing as in the classical case, these situations would trigger a change in the prices of the underlying assets which will make such opportunities disappear. We stress that $\sNA$ does not imply that the classical no arbitrage condition holds under any of the $P\in\cP$.

As pointed out above, if $\cP$ is dominated, we clearly have $\tildeC{}=\cC$ (i.e. $\cC$ is sensitive). However, as we demonstrate below this is not always the case under Knightian uncertainty, unless the framework is chosen carefully.
We show that the discrepancy $\tildeC{}\neq\cC$ can be often resolved by extending the filtration in an appropriate technical way, which allows for aggregation of $P$-dependent strategies (see Section \ref{sec:C}) and for which $\sNA\Leftrightarrow\NA$.  Since the aim of the paper is to provide a general FTAP which is not tailor made to any specific underlying setting we refrain to assume that $\cC$ is sensitive and continue to work with $\tildeC{}$.

\paragraph{The sensitive version of the FTAP.}

We first introduce the class of dual elements. 
Recall that $\cC$ and $\CM$ are defined in \eqref{cone:abstract} and \eqref{cone:abstract:M} for $\lambda\ge 0$.

\begin{definition}An \emph{approximate} separating class is a sequence of probabilities $\cQ:=\{Q_n\}_{n\in\N}$ such that there exists $P\in\cP$ with $\cQ\ll P$ and, for any $n\in\N$,
	\begin{equation}\label{eq:sepCN}
		E_{Q_{n}}[X]\leq
		\frac{1}{n}\quad \forall\,X\in \cC_{n},
	\end{equation}
	We denote by $\Qapp$ the collection of approximate separating classes.
\end{definition}

We now state the main result of the section.
To this end recall that $\polar$ represents the class of polar sets for $\cP$ and that a set $\cA\subset \LW$ is Fatou-closed 
if for any $\normW{\cdot}$-bounded sequence $\{X_n\}_{n\in\N}\subset\cA$, $X_n \rightarrow X$ $\qs$, we have $X\in\cA$. 
\begin{theorem}\label{FTAP}
	Suppose that under $\sNFLVR$ $\cC$ is Fatou closed. The following are equivalent:
	
	\begin{enumerate}
		
		\item $\sNFLVR$
		
		\item $\Qapp\approx\cP$ \footnote{With a slight abuse of notation, $\Qapp\approx \cP$ means that the whole collection of probabilities belonging to some approximate separating class is equivalent to $\cP$.}. Moreover, $\forall A\in\cF\setminus\polar$, $\exists \delta>0,\ \cQ\in\Qapp$ such that $\inf_{Q\in\cQ}Q(A)\ge\delta$.
	\end{enumerate}
\end{theorem}
For discrete time financial market models, $\sNFLVR$ guarantees that $\cC$ is Fatou closed (see Lemma \ref{closure} below), as a consequence, we do not need such an assumption in the subsequent Theorem \ref{FTAP_discrete} and \ref{FTAP_compact}. 
Whether the same result holds for continuous time models is an interesting question which goes beyond the scope of this paper and is left for  future investigations. Interestingly, in the recent paper \cite{CKPS19}, it is shown that a general MOT duality holds \emph{only} if the set of attainable payoffs is Fatou closed.

We provide the proof of Theorem \ref{FTAP} in Section \ref{sec:proofs}. 
One of the main technical point is to show that the sensitive version of any $\CM$ is closed in an appropriate weak topology, which is proven in Proposition \ref{closed:sensitive}.
In concrete models, $\CM$ contains a sufficiently rich class of dynamic strategies in the underlying process $S$ which allows to identify martingale measures. The above theorem essentially says that $\sNFLVR$ is equivalent to the fact that for any non-polar set $A$, we can find \emph{approximate} martingale measures for $S$ which assign positive probability to $A$ (see also \eqref{eq:apprMM} below).
In particular this implies that, under $\sNFLVR$, the class of approximate martingale measures is non empty and equivalent to $\cP$.

\begin{remark}
	Given $\cQ\in\Qapp$, it is possible to define a super-additive functional $\psi(\cdot):=\inf_{Q\in\cQ}E_Q[\cdot]$ in the spirit of \cite{ATY01} and \cite{AT02}, which, by \eqref{eq:sepCN}, is a non-linear separator of the cone $\cC$.
	In the context of Knightian Uncertainty, non-linearity arises also for pricing rules related to economic equilibrium and absence of arbitrage (see e.g. \cite{BR19,BRS18}).
\end{remark}

We call Theorem \ref{FTAP} an ``abstract'' version of the FTAP since it is obtained in a general setup and its implications can be strenghtened if we are willing to choose a more specific setting or adopt stronger assumptions.
More precisely:
\begin{enumerate}
	\item From a technical point of view, a desirable property is $\tildeC{}=\cC$ which can hold under the $\NA$ assumption (resp. $\NFLVR$), provided a suitable underlying structure. By definition, $\tildeC{}=\cC$ automatically implies $\NA \Leftrightarrow \sNA$ (resp. $\sNFLVR \Leftrightarrow \NFLVR$).
	As discussed above, such a situation occurs when $\cP$ is dominated.
	We will explain in Section \ref{sec:C} that this is related to the choice of the filtration and it holds true in the framework of \cite{BN13}, where, in addition, these four notions of arbitrage are all equivalent.
	\item The approximate separating classes of Theorem \ref{FTAP} can be used to obtain linear pricing functionals as ``true'' martingale measures for a given underlying process. 
	In Section \ref{discrete:setting} we will show that this is possible under some further assumptions and illustrate the result in a discrete-time MOT framework.
	\item When both the two points (1. and 2.) above are fulfilled, Theorem \ref{FTAP} in discrete time has the more familiar form:
	\[\NA\Longleftrightarrow\cQ_{mtg}\approx\cP,
	\]
	where $\cQ_{mtg}:=\{Q\lll \cP \mid E_{Q}[k]=0 \;\forall\, k\in \cK\}$.
	\item (On No Free Lunch) Mathematically, one could obtain separating measures using the Hahn Banach Theorem, under the No Free Lunch (NFL) condition: $\clw(\Ccal)\cap
	\Linfplus=\{0\},$
	where $\clw$ denotes the closure with respect to $\sigma(\Linf,ca(\cP))$-topology of $\cC$.
	As in the dominated case, it is not a priori clear how limit points in $\clw(\Ccal)$ are related to the payoffs of implementable strategies, thus, a clear economic interpretation is missing.
\end{enumerate}

\section{Discrete-time and martingale measures}\label{discrete:setting}
In this section we further analyze the discrete-time setting and show how to obtain martingale measures from Theorem \ref{FTAP}.
Let $T\in
\mathbb{N}$, and $I:=\left\{ 0,...,T\right\} $. The
price process is given by an 
$\mathbb{R}^{d}$-valued stochastic process $S=(S_{t})_{t\in
	I}$ with $S_t^j\in \Lzero$  for every $t\in I, j=1,\ldots,d$, and we
also assume the existence of a numeraire asset $S_{t}^{0}=1$ for
all $ t\in I$. Moreover, we fix a filtration
$\mathbb{F}:=\{\mathcal{F}_{t}\}_{t\in I}$ such that the process
$S$ is $\mathbb{F}$-adapted. 
A finite set of $\cF$-measurable options $\Phi=(\phi_1,\ldots,\phi_m)$ ($\phi_i\in \Lzero$ for every $i$) is available for static trading and, without loss of generality, we assume their initial price to be $0$.
An admissible semi-static strategy is a couple $(H,h)$ where $H$ is an $\R^d$-valued, $\mathbb{F}$-predictable stochastic process with $H_t^j\in \Lzero$  for every $t\in I, j=1,\ldots,d$ and 
$h\in \R^m$. The final payoff is $(H\circ S)_T+ h\cdot\Phi\in \Lzero$ where the stochastic integral is defined as 
$$(H\circ
S)_t:=\sum_{k=1}^{t}\sum_{j=1}^{d}H_{k}^{j}(S_{k}^{j}-S_{k-1}^{j}),\quad t\in\cI,
$$
with $(H\circ S)_0=0$.
We denote by $\Hcal$ the class of semi-static admissible strategies with zero initial cost. 

We choose
\[ W:=1+\sum_{t=1}^T\sum_{j=1}^d |S^j_t|+\sum_{i=1}^m |\phi_i|.
\] 
The sets $\cC$ and $\CM$ takes the following explicit form
\begin{eqnarray}
	\Ccal&=&\left\{ X\in \LW\mid X\leq (H\circ S)_T+h\cdot \Phi \text{ for some } (H,h)\in
	\Hcal\right\},\label{cone:option}\\
	\CM&=&\left\{ X\in \LW\mid X\leq (H\circ S)_T+h\cdot \Phi \text{ for some } (H,h)\in
	\HM\right\},\label{cone:option:M}
\end{eqnarray}
where $\HM := \{(H,h)\in
\Hcal \mid (H\circ S)_T+h\cdot \Phi \geq -\lambda W\}$ for $\lambda\ge 0$.
Note that, in particular, $\cup_{\lambda\ge 0}\HM$ contains any bounded strategy.

Recall that a set $\cA\subset \Lzero$ is closed with respect to $\qs$ convergence if for any sequence $\{X_n\}_{n\in\N}$, $X_n \rightarrow X$ $\qs$ implies $X\in\cA$.

\begin{lemma}\label{closure} Under $\NA$ the convex set $\cC$ is closed with respect ot $\qs$ convergence and hence both $\cC$ and $\CM$ are Fatou closed. 
\end{lemma}
\begin{proof}
	A direct application of \cite[Remark 2.1 and Theorem 2.2]{BN13} guarantees that  $\Ccal$ is closed with respect to $\qs$ convergence.
	Since $\cC$ is closed with respect to $\qs$ convergence, it is Fatou closed.
	Consider now a $\normW{\cdot}$-bounded sequence
	$\{X_n\}\subset \CM$ such that $X_n\to X$, $\qs$ for some $X\in\LW$.
	By definition of $\CM$, there exists $(H_n,h_n)\in \HM$ such that
	$X_n\leq (H_n\circ S)_T+h_n\cdot \Phi$, from which we deduce that $(-\lambda W)\vee X_n \in\CM\subset \cC$.
	Moreover, from the closure of $\cC$ with respect to $\qs$ convergence, the limit $(-\lambda W)\vee X=:\tilde{X}$ belongs to $\cC$.
	By definition of $\cC$, there exists $(H,h)\in \Hcal$ such that $\tilde{X}\leq (H\circ S)_T+h\cdot\Phi$ and, necessarily, $(H,h)\in
	\HM$. From $X\le\tilde{X}$ and the monotonicity of $\CM$, $X\in\CM$. 
\end{proof}

Using Lemma \ref{closure} we can specialize Theorem \ref{FTAP} to the present discrete time setup without the Fatou closure assumption.

\begin{theorem} \label{FTAP_discrete}
	The following are equivalent:
	\begin{enumerate}
		
		\item $\sNFLVR$
		
		\item $\Qapp\approx\cP$. Moreover, $\forall A\in\cF\setminus\polar$, $\exists \delta>0,\ \cQ\in\Qapp$ such that $\inf_{Q\in\cQ}Q(A)\ge\delta$.
	\end{enumerate}
	The two conditions are further equivalent to $\NA$ if, in addition, $\cC$ is sensitive under $\NA$.
\end{theorem}

We call measures in $\Qapp$ \emph{approximate martingale measures}.
Indeed, for every $A\in\cF_{k-1}$ the one-step strategy $H=(H_t)_{t\in I}$ with $H_t(\omega)=1_A(\omega)1_{\{k\}}(t)$ for every $\omega\in\Omega$ and $t\in I$, satisfies $(H,0)\in\HM$.
Similarly $(0,\pm e_i)\in\HM$ for every $i=1,\ldots,m$, where $\{e_i\}_{i=1}^m$ denotes the canonical basis of $\R^m$.
Their final payoffs are thus contained in $\CM$, from which, for every $n\in\N$ and $\{\QN\}\in\Qapp$, it holds
\begin{equation}\label{eq:apprMM}
	|E_{\QN}[\phi_i]| \le \frac{1}{n},
	\qquad
	|E_{\QN}[1_A(S^j_k-S^j_{k-1})]| \le \frac{1}{n},
\end{equation}
for any $i=1,\ldots,m$, $j=1,\ldots,d$, $A\in\cF_{k-1}$, $k=1,\ldots,T$.

We now show that, under some additional weak assumptions, Theorem \ref{FTAP} implies the existence of true martingale measures.

\begin{assumption}\label{assumption} $(\Omega,m)$ is a Polish space with respect to a metric $m$.
	\begin{enumerate}[(i)]
		\item\label{hp:cont} For any $t\in\cI$, $S_t:\Omega\to\R^d_+$ is a continuous function \footnote{More precisely is $\qs$ equal to a continuous function};
		\item \label{hp:compact} For any $P\in \cP$ there exists a compact set $K^P$ such that $P(K^P)=1$. 
		\item $\mathbb{F}:=\{\cF_t\}_{t\in\cI}$ is the natural filtration generated by $S$.
	\end{enumerate} 
\end{assumption} 

Note that the previous conditions are not restrictive. If $S$ is only Borel measurable, by \cite[Theorem 4.59]{AB06} there exists a Polish topology $\tau$ on $\Omega$ such that the Borel sigma algebra is the same and the process $S$ is $\tau$-continuous.
Thus, assumption \eqref{hp:cont} can be made without loss of generality.
Assumption \eqref{hp:compact} can be easily fulfilled when the class of priors $\cP$ has the only scope of fixing the polar sets.
For $\Omega$ a Polish space, the class $\cR:=\{P(\cdot\mid K)\mid K\subset\Omega\text{ compact},\ P\in\cP\}$ satisfies \eqref{hp:compact} and $\cR\approx\cP$. 
Indeed, $\cR\ll \cP$ is trivial.
If $A\in\cF\setminus\polar$, there exists $P\in\cP$ such that $P(A)>0$. By \cite[Theorem 12.7]{AB06}, we find a compact set $K\subset A$ such that $P(K)>0$, from which $P(A\mid K)>0$ and $\cP\ll\cR$.

Let now $\mtg :=\{Q\in\overline{\cQ}\mid \cQ\in\Qapp,\ S \text{ is a $(Q,\mathbb{F})$-martingale with }E_Q[\Phi]=0\}$,
where the closure $\overline{\cQ}$ is taken in the $\sigma(\cM_1,\cC_b)$ sense.

\begin{theorem}\label{FTAP_compact}
	Under Assumption \ref{assumption}, the following are equivalent:
	\begin{enumerate}
		\item $\sNFLVR$;
		\item $\Qapp\approx\cP$ and $\cP\ll\mtg$.
	\end{enumerate}
\end{theorem}

\begin{remark}
	Theorem \ref{FTAP_compact} cannot guarantee that the limiting measures $\mtg$ satisfies $\mtg\ll \cP $, as weak limits do not, in general, preserve absolute continuity with respect to a measure. 
\end{remark}

\paragraph{A martingale optimal transport framework.}
Set $\Omega=\R^{d\times T}_+$ and let $S_t(\omega)=\omega_t$ be the canonical process.
We assume that, for any of the assets $S^j$, a certain finite number $N(j)$ of call options are available for semi-static trading with payoffs $(S^j_T-k^j_i)_+$ for some $k^j_i>0$ and with prices $c^j_i$, for $i=1,\ldots N(j)$.
We assume that $c^j_i\ge 0$, otherwise there is an obvious arbitrage opportunity, and we also assume that for a sufficiently large strike price the options are traded at zero price; we model this by setting $c^j_{N(j)}=0$.
The corresponding set of options with zero prices is given by $\Phi=\{(S^j_T-k^j_i)_+-c_j\mid j=1,\ldots,d,i=1,\ldots,N(j)\}$. Following \cite{DH07} we construct the support function $\cR^j$ as the maximal convex non-increasing function such that $\cR^j(k^j_i)\le c^j_i$. As $\cR^j$ is 
$\lambda$-a.s.\footnote{$\lambda$ is the Lebesgue measure on $\R$.}~twice differentiable, following the observation of \cite{BL78}, we define a probability measure $\mu^j$ on $\R$ as $d\mu^j/d\lambda=(\cR^j)''$.
Note that, from the assumption $c^j_{N(j)}=0$, it follows that $\cR^j(x)=0$ for all $x\ge k^j_{N(j)}$ and, therefore, $\mu^j$ has compact support.

Let $\mu:=\otimes_{j=1}^d \mu^j$ the product measure of $\{\mu^j\}_{j=1}^d$ on $\R^d$. Let $K$ be the compact support of $\mu$. 
We consider a family of probability measures $\cP$ satisfying
\begin{equation}
	\cP\subset\{P\in\cM_1\mid P_T\sim \mu\},
\end{equation}
where $P_T$ denotes the marginal of $P$ on the last component of $\Omega$. The interpretation is the following. If we denote by $\cQ$ the (unknown) set of measures which are used in the market to price the options $\Phi$, $\mu$ represents the (approximation) of the distribution of $S_T$ under any $Q\in\cQ$.
Any probability measure equivalent to $\cQ$ defines the same null-events and should be considered as a plausible model.
Therefore, the only constraint that we can deduce from market data is that the distribution of $S_T$ under $P\in\cP$ should be equivalent to $\mu$. Note that, differently from the standard martingale optimal transport setup, we are not assuming to know, in addition, all the marginals at intermediate time. This case could be easily incorporated.\\

Denote by $K^T$ the $T$-fold product of the compact set $K$.
\begin{lemma}\label{l.arb}
	Under $\NA$, $P(K^T)=1$ for every $P\in\cP$.
\end{lemma}
\begin{proof}
	For every $n\in\N$, $t=1,\ldots,T-1$, consider the closed-valued multifunction 
	\[
	\Psi_{t,n}(\omega):=\left\{H\in\R^d\mid H\cdot(x-\omega_t)\ge \frac{1}{n}\quad \forall x\in K\setminus\{\omega_t\}\right\},\quad \omega\in\Omega.
	\]
	The domain of $\Psi_{t,n}$ is defined as $\textrm{dom}(\Psi_{t,n}):=\{\omega\in\Omega\mid\Psi_{t,n}(\omega)\neq\emptyset\}$.
	The compactness of $K$ and the hyperplane separating theorem implies that $\cup_{n\in \N}\textrm{dom}(\Psi_{t,n})=\{\omega\in\Omega\mid S_t(\omega)\notin K\}$. 
	From \cite[Lemma A.7]{BFM16}, $\Psi_{t,n}$ is $\cF^S_t$-measurable, thus, it admits a measurable selector $\psi_{t,n}$ on its domain which we extend to the whole $\Omega$ by setting $\psi_{t,n}=0$ on the complementary set.
	By letting $H_t:=\sum_{n=1}^{\infty}\psi_{t,n}1_{\{\psi_{t,n-1}\neq 0\}}$ with $\psi_{t,0}=0$, we obtain $H_t\cdot(S_T-S_t)\ge 0$ with strict positivity on $\{\omega\in\Omega\mid S_t(\omega)\notin K\}$. If now, by contradiction, $P(K^T)<1$, there exists $1\le t\le T-1$ such that $P(S_t\in K^c)>0$.
	Thus, $H_t$ as above provides an arbitrage opportunity.
\end{proof}

We could deduce the following
\begin{proposition}
	Under the assumption of this paragraph, the following are equivalent:
	\begin{enumerate}
		\item $\sNFLVR$;
		\item $\Qapp\approx\cP$ and $\cP\ll\mtg$.
	\end{enumerate}
\end{proposition}
\begin{proof}
	Assumption \ref{assumption} is satisfied in the framework of this subsection. The result follows directly from Theorem \ref{FTAP_compact}. 
\end{proof}

\begin{remark}
	If, in addition, the class $\cP$ is chosen with the structure of \cite{BN13}, from Theorem \ref{example:sensitive} below, the above are further equivalent to $\NA$.
\end{remark}
\begin{remark}
	In the classical MOT framework it is well know that Strassen's Theorem ensure that the set of martingale measures with prescribed marginals is non empty if and only if the marginals are in convex order. This is typically taken as a no-arbitrage condition. The above theorem explains such a no-arbitrage condition from a different point of view.
\end{remark}

\section{The role of filtrations in the aggregation process}\label{sec:C}

In this section we depict two well known examples, borrowed from the recent literature, in which the cone $\Ccal$ turns out to be sensitive under the no-arbitrage condition. We show how sensitivity is related to the possibility of obtaining an aggregation property for superhedging strategies. In both examples the filtration will play a crucial role and will be an opportune enlargement of the natural filtration, which will not affect the structure of the set of martingale measures for the discounted price process, calibrated on liquid options. We stress that our main goal is to explain some significant features of the models rather than recovering these well established results. Nevertheless, Theorem \ref{example:sensitive} below is a direct proof of the sensitivity of $\cC$ which does not rely on the results of \cite{BN13} and which provides new insights on the properties of the superhedging functional in that framework.

\paragraph{The product structure of \cite{BN13}.}
Starting from the framework of Section \ref{discrete:setting} and letting $W=1$, we further assume the following requirements. The underlying space $\Omega=\Omega_1^T$ is a $T$-fold product of a Polish space $\Omega_1$ and $\Omega_t:=\Omega_1^t$. 
For every $t\in I$, $\cF_t$ is the universal completion of the Borel sigma-algebra $\cB_{\Omega_t}$, defined as
\begin{equation*}
	\bigcap_{P\in \cM_1(\Omega_t)}\mathcal{B}_{\Omega_t}\vee \mathcal{N}_t^{P},
	\text{ where }\mathcal{N}_t^{P}=\{N\subset A\in \mathcal{B}_{\Omega_t}\mid P(A)=0\}\text{.}
\end{equation*}
Fix $t\in I$, the event $\omega\in \Omega_t$ can be seen as a path observed up to time $t$ and  $\cP_t(\omega)\subset \cM_1(\Omega_1)$ is a prescribed convex set of priors, on the node $(t,\omega)$. It is assumed that
$$\text{graph}(\cP_t)=\{(\omega,P)\mid \omega\in \Omega_t,\; P\in \cP_t(\omega)\}$$
is analytic, thus, it admits a universally measurable selector $P_t:\Omega_t\to \cM_1(\Omega_1)$: this allows to introduce the set of multiperiod probabilistic models (priors) as 
$$\cP := \{P_0\otimes P_1\otimes \ldots \otimes P_{T-1}\mid P_t(\cdot)\in \cP_t(\cdot),\, t=0,\ldots,T-1\}. $$
We set 
\begin{equation*}
	\mathcal{Q}:= \left\{ Q\lll \cP\mid S\text{ is an }\mathbb{F}\text{-martingale under }Q \text{ and } E_{Q}[\Phi]=0\right\}.
\end{equation*}
For simplicity we assume in the following that $\Phi=0$.  
In this specific framework the following FTAP was proved in \cite{BN13}.

\begin{center} $\NA$ holds if and only if 
	$\cP$ and $\mathcal{Q}$ share the same polar sets $\polar$.\end{center}

Our aim is to establish the sensitivity of the cone $\Ccal$ by showing that the $\cP$-$\qs$ superhedging price is the supremum of the $P$-a.s.\ superhedging price with $P$ varying in $\cP$. 

\begin{theorem}\label{example:sensitive} Consider the measurable space $(\Omega,\cF_T)$ and the class $\cP$ as described above. Then under $\NA$ we have 
	\[ \Ccal = \bigcap_{P\in\cP} j_P^{-1}(j_P(\Ccal)) \]
	and the four notions $\NA$, $\sNA$, $\NFLVR$ and $\sNFLVR$ are all equivalent. 
\end{theorem}

Before giving the proof we need some preliminary results.
Denote by $\Pi(X)$ the quasi-sure superhedging price of an upper semianalytic function $X:\Omega\to\R$ (see \cite[Chapter 7.7]{BS78} for more on semianalytic functions).
For $t=0,\ldots, T-1$, let $\Delta S_{t+1}:=S_{t+1}-S_t$ and for a function $f:\Omega_t\times\Omega_1\to \R^n$, the notation $f(\omega;\cdot)$ indicates that $f$ is considered as a function on $\Omega_1$ with the first coordinate fixed.
Let $\Pi_T=X$ and define, by backward iteration,
\[\Pi_t(\omega):=\inf\{x\in\R\mid \exists H\in\R^d \text{ s.t. } x+H\cdot \Delta S_{t+1}(\omega;\cdot)\ge \Pi_{t+1}(\omega;\cdot)\  \cP_t(\omega)\text{-}\qs\},
\]
for $\omega\in\Omega_t$.
We refer to $\Pi_t$ as the conditional superhedging price of $X$ at time $t$, in particular $\Pi_0=\Pi(X)$.	
In the same spirit, we define  
\[\pi_t(\omega,P):=\inf\{x\in\R\mid \exists H\in\R^d \text{ s.t. } x+H\cdot \Delta S_{t+1}(\omega;\cdot)\ge \Pi_{t+1}(\omega;\cdot)\  P\text{-a.s.}\},
\]
for $\omega\in\Omega_t$.
We aim at showing that under $\NA$ there exists a $\cP$-polar set $N$ such that the following holds:
\begin{enumerate}
	\item $\Pi_t(\omega)=\sup_{P\in\cP_t(\omega)}\pi_t(\omega,P)$, for every $\omega\in N^c$;
	\item $\Pi_t:\Omega_t\to \overline{\R}$ is upper semianalytic;
	\item for any $\varepsilon>0$, there exists a universally measurable kernel $P^\varepsilon:\Omega_t\to\cM_1(\Omega_1)$ such that $P^\varepsilon(\omega)\in\cP_t(\omega)$ and $\pi_t(\omega,P^\varepsilon_t(\omega))\ge \Pi_t(\omega)-\varepsilon$,  for every $\omega\in N^c$.
\end{enumerate}
The first property shows that the quasi-sure superhedging price for a given $X$ is the worst case among all the $P$-a.s superhedging prices for $X$. 
The other two properties guarantee that it is possible to construct models in $\cP$, by means of an appropriate measurable selection of $\cP_t$, for which the almost-sure superhedging price is arbitrarily close to $\Pi(X)$.

\begin{remark}\label{rmk:rewrite}
	Let $\cP\subset\cM_1$ and $X\in\Lzero$. For any $P\in\cP$, the measure $\frac{dP'}{dP}=\frac{c}{1+|X|}$, with $c$ a normalizing constant, is equivalent to $P$ and $E_{P'}[|X|]<\infty$. Thus, $\cP\approx \tilde{\cP}$ where $\tilde{\cP}=\{P\lll\cP\mid E_P[|X|]<\infty\}$. Now we show that
	\[
	\inf\{x\in\R\mid x\ge X\ \cP\text{-}\qs\}=\sup_{P\in \tilde{\cP}}E_{P} [X].
	\]
	Indeed, if the l.h.s.\ is infinite the inequality $\ge$ is trivial. Otherwise, if $x$ is such that $x\ge X$ $\cP$ $\qs$ then $x\ge E_{P}[X]$ for any $P\ll P'$ with $E_P[|X|]<\infty$ and $P'\in\cP$; the latter exists from the first observation above. The inequality $\ge$ follows.
	Let $M$ be the value of the l.h.s.\ above (or an arbitrary large $M$ if it is not finite). For any $\varepsilon>0$, there exists $P'\in\cP$ such that the set
	$A:=\{M-\varepsilon< X\}$ satisfies $P'(A)>0$. Take $P''\sim P'$ such that $E_{P''}[|X|]<\infty$ as above and note that it still holds $P''(A)>0$. Define the probability $P(\cdot):=P''(\cdot\mid A)$ which satisfies $P\ll P''\sim P'$ and $ M-\varepsilon<E_{P}[X]$. From $\varepsilon$ (and $M$ in the infinite case) being arbitrary, the inequality $\le$ follows.
\end{remark}
We start by showing the measurability properties.
\begin{proposition}\label{prop:USA}
	For any $t\in I$, $\Pi_t:\Omega_t\to\overline{\R}$ is upper semianalytic.
\end{proposition}
\begin{proof}
	We proceed by backward iteration. 
	For $t=T$, $\Pi_T=X$ which is upper semianalytic by assumption.
	Assume the same is true up to $t+1$.
	Using Remark \ref{rmk:rewrite}, we first rewrite the conditional superhedging price as follows, 
	\begin{eqnarray*}
		\Pi_t(\omega)& = &\inf_{H\in\R^d}\inf\{x\in\R\mid x\ge \Pi_{t+1}(\omega;\cdot)-H\cdot \Delta S_{t+1}(\omega;\cdot)\quad \cP_t(\omega)\text{-}\qs\}
		\\ & = & \inf_{H\in\R^d}\sup_{P\in\tilde{\cP}_t(\omega)}E_P [\Pi_{t+1}(\omega;\cdot)-H\cdot\Delta S_{t+1}(\omega;\cdot)],
	\end{eqnarray*}
	where $\tilde{\cP}_t(\omega):=\{P\lll\cP_t(\omega) \mid E_P[|\Delta S_{t+1}(\omega;\cdot)|]<\infty\}$ and where we recall that $E_P [\Pi_{t+1}(\omega;\cdot)]=\int_{\Omega_1}\Pi_{t+1}(\omega;\tilde{\omega})dP(\tilde{\omega})$ (similarly for the other term).
	By an application of the minimax theorem (see e.g.\ \cite[Corollary 2]{T72}), we can rewrite
	\begin{eqnarray*}
		\Pi_t(\omega)&=&\inf_{n\in\N}\inf_{H\in B_n(0)}\sup_{P\in\tilde{\cP}_t(\omega)}E_P [\Pi_{t+1}(\omega;\cdot)-H\cdot\Delta S_{t+1}(\omega;\cdot)]\\&=&\inf_{n\in\N}\sup_{P\in\tilde{\cP}_t(\omega)}\inf_{H\in B_n(0)}E_P [\Pi_{t+1}(\omega;\cdot)-H\cdot\Delta S_{t+1}(\omega;\cdot)],
	\end{eqnarray*}
	where $B_n(0)$ is the closed (compact) ball in $\R^d$ of radius $n$ centered in $0$. 
	Indeed, the above function of $(H,P)$ is affine in the first variable (hence convex and continuous) and linear (hence concave) in the second one.
	We now show that the functions $f_n:\Omega_t\times\cM_1(\Omega_1)\to \overline{\R}$ defined as 
	\begin{eqnarray}\label{eq:USA}
		f_n(\omega,P)& = & E_P[\Pi_{t+1}(\omega;\cdot)]+\inf_{H\in B_n(0)}g(\omega,P,H),
		\\ g(\omega,P,H)& := & -H\cdot E_P[\Delta S_{t+1}(\omega;\cdot)],\nonumber
	\end{eqnarray}
	are upper semianalytic.
	We first observe that $g$ is continuous in $H$ with $(\omega,P)$ fixed and $g$ is Borel measurable in $(\omega,P)$ with $H$ fixed, as a consequence of \cite[Proposition 7.26 and 7.29]{BS78}.
	In other words $g$ is a Carath\'eodory function, which is a particular case of normal integrand (see \cite[Definition 14.27]{R})
	By \cite[Theorem 14.37]{R}, $(\omega,P)\mapsto\inf_{H\in B_n(0)}g(\omega,P,H)$ is Borel measurable (in particular upper semianalytic).
	Moreover, since $\Pi_{t+1}$ is upper semianalytic by the inductive assumption, the map $(\omega,P)\mapsto E_{P}[\Pi_{t+1}(\omega;\cdot)]$ is also upper semianalytic by \cite[Proposition 7.26 and 7.48]{BS78}.
	We conclude that $f_n$ is upper semianalytic as the sum of upper semianalytic functions.
	
	Define now the Borel-measurable function $\phi:\Omega\times \cM_1(\Omega_1)\times \cM_1(\Omega_1)\mapsto \R\cup \{+\infty\}$ as $\phi(\omega,P,P'):=E_P[\frac{dP'}{dP}]$ if $P'\ll P$ and $+\infty$ otherwise. 
	Consider the set $D\subset \Omega_t\times \cM_1(\Omega_1)$ defined by
	\[D:=\{(\omega,P)\mid \omega\in \Omega_t, \; P \lll \cP_t(\omega), \; E_P[|\Delta S_{t+1}(\omega;\cdot)|]<\infty \}.
	\]
	From the definition of $\phi$ we have $\phi(\omega,P,P')=1$ if and only if $P'\ll P$ and therefore $D$ is the projection on the first and the third components of the analytic set
	$\big(\graph(\cP_t)\times\cM_1(\Omega_1)\big)\cap \phi^{-1}(1)$ intersected with the Borel set of probability measures for which $\Delta S_{t+1}$ is integrable.
	We deduce that $D$ is  also analytic (see also \cite[Lemma 4.11]{BN13}).
	
	Clearly $\sup_{P\lll\cP_t(\omega)} f_n(\omega,P)=\sup_{P\in D_{\omega}} f_n(\omega,P)$ with $D_{\omega}=\{P\in \cM_1(\Omega_1)\mid (\omega,P)\in D\}$ and by \cite[Proposition 7.47]{BS78}, $\omega\mapsto\sup_{P\lll\cP_t(\omega)} f_n(\omega,P)$ is upper semianalytic. Finally, the class of upper semianalytic functions is closed under countable infimum by \cite[Lemma 7.30 (2)]{BS78} and the result follows.
\end{proof}

In exactly the same way we can show the following.

\begin{lemma}\label{lem:Palmost}
	The map $\psi:\Omega\to\overline{R}$, defined as $\psi(\omega)=\sup_{P\in\cP_t(\omega)}\pi_t(\omega,P)$ is upper semianalytic.
	Moreover, for any $\varepsilon>0$, there exists a universally measurable kernel $P^\varepsilon_t:\Omega_t\to\cM_1(\Omega_1)$ such that $P^\varepsilon_t(\omega)\in\cP_t(\omega)$ and $\pi_t(\omega,P^\varepsilon_t(\omega))\ge  \psi(\omega)-\varepsilon$ on $\psi<\infty$ and $\pi_t(\omega,P^\varepsilon_t(\omega))\ge  1/\varepsilon$ on $\psi=\infty$.
\end{lemma}
\begin{proof}
	As in the proof of Proposition \ref{prop:USA}, the function $f_n$ in \eqref{eq:USA} is upper semianalytic for every $n\in\N$.
	Note that $\pi_t(\omega,P)=\inf_{n\in\N}\sup_{P'\ll P} f_n(\omega,P)$ so that, $\pi_t$ is again upper semianalytic.
	Using the fact that $\graph(\cP_t)$ is analytic we deduce that $\psi$ is upper semianalytic and the existence of a universally measurable $\varepsilon$-optimizer from \cite[Proposition 7.50]{BS78}.
\end{proof}

Given the above measurability properties we can now focus on the pointwise property 1.\ above.
\begin{proposition}\label{prop:supsuper}
	Assume $\NA$ and $X\in \Linf$.
	There exists a $\cP$-polar set $N$ such that $\Pi_t(\omega)=\sup_{P\in\cP_t(\omega)}\pi_t(\omega,P)$, for every $t\in I$ and $\omega\in N^c$.
\end{proposition}
\begin{proof}
	The inequality $\ge$ is trivial so we only need to show the converse. From $\NA$, the local condition $NA(\cP_t(\omega))$ holds for any $\omega$ outside a polar set $N$ (see \cite[Lemma 4.6]{BN13}) and we can assume $\Pi_t(\omega)<\infty$ as $X$ is $\qs$ bounded.
	Fix $\omega\in N^c$ and $t\in I$.  
	If $\Pi_t(\omega)=-\infty$ the equality follows so we also assume $\Pi_t(\omega)>-\infty$.
	
	Without loss of generality we assume that the random vector $S_{t+1}$ is composed only of non-redundant asset. Let 
	$$\check{\Pi}_t(\omega)=\sup\{x\in\R\mid \exists H\in\R^d \text{ s.t. } x+H\cdot \Delta S_{t+1}(\omega;\cdot)\le \Pi_{t+1}(\omega;\cdot)\ \cP_t(\omega)\text{-}\qs\},$$ 
	for $\omega\in\Omega_t$, be the subhedging price. We can distinguish two cases: $\check{\Pi}_t(\omega)<\Pi_t(\omega)$ or  $\check{\Pi}_t(\omega)=\Pi_t(\omega)$.
	In the first case, for any $x\in (\check{\Pi}_t(\omega),\Pi_t(\omega))$ we have that the local no arbitrage condition holds for the extended one period market $Y_0=0$, $Y_1:=[\Delta S_{t+1}; -\Pi_{t+1}+x]$. Indeed, suppose $[H,h]\in\R^{d+1}$ satisfies $H\cdot \Delta S_{t+1}+ h(-\Pi_{t+1}+x)\ge 0$. If $h= 0$, $NA(\cP_t(\omega))$ and the non-redundancy imply $H=0$.
	If $h\neq 0$, dividing both sides by $h$ we see that $H/h$ is either a sub or a super-hedging strategy which is not possible by the choice of $x$. 
	Since $NA(\cP_t(\omega))$ holds for the one-period market, \cite[Lemma 2.7]{BZ17} implies that there exists $P'\in\cP_t(\omega)$ such that $NA(P')$ holds and for which the assets are still non redundant.
	For the same reason as above, a $P'$-a.s. superhedging strategy with initial price $x$ would be a $P$-a.s. arbitrage in the extended market, which is excluded. Thus, $\pi_t(\omega,P')\ge x$ and  $\sup_{P\in\cP_t(\omega)}\pi_t(\omega,P)\ge x$. By taking the supremum for $x\in (\check{\Pi}_t(\omega),\Pi_t(\omega))$, the desired inequality follows.
	
	Finally, when $\check{\Pi}_t(\omega)=\Pi_t(\omega)<\infty$, we have that  for some $H,\check{H}\in\R^{d}$,
	\[\Pi_t(\omega)+H\cdot\Delta S_{t+1}\ge\Pi_{t+1}\ge\check{\Pi}_{t+1}\ge\Pi_t(\omega)+\check{H}\cdot\Delta S_{t+1},\quad \cP_t(\omega)\text{-}\qs. \]
	By the local no arbitrage condition $H=\check{H}$ and $\Pi_{t+1}$ is $\cP_t(\omega)$-$\qs$ replicable.
	Similarly as before the local no arbitrage condition holds for the extended one period market $(Y_0,Y_1)$. Applying again \cite[Lemma 2.7]{BZ17} we deduce that there exists $P'\in\cP_t(\omega)$  such that 
	$\pi_t(\omega,P')=\Pi_t(\omega)$. This implies the desired inequality. 
\end{proof}

We can now conclude the proof of the sensitivity of $\cC$.

\begin{proof}[Proof of Theorem \ref{example:sensitive}]
	Let $X\in\tilde{\cC}$, namely, the $P$-a.s. superhedging price $\Pi^P(X)$ is non-positive, for any $P\in\cP$.
	For any $\varepsilon>0$, consider the probability measure
	\[P^\varepsilon:=P^\varepsilon_1\otimes\cdots \otimes P^\varepsilon_t\otimes\cdots \otimes P^\varepsilon_T,
	\]
	constructed via Fubini's Theorem from the universally measurable kernels $\{P_t^\varepsilon\}$ from Lemma \ref{lem:Palmost}.
	By construction $P^\varepsilon\in\cP$ and, by Proposition \ref{prop:supsuper}, $\Pi(X)\le \Pi^{P^\varepsilon}(X)+\varepsilon T$.
	As $X\in\tilde{\cC}$, $\Pi^{P^\varepsilon}(X)\le 0$ and since $\varepsilon$ is arbitrary, we conclude $\Pi(X)\le 0$.
\end{proof}
\begin{remark} The previous results reads as follows: fix $X\in \Linf$ and assume that for every $P\in\mathcal{P}$
	we find $H^P\in \Hcal$ such that $g\le (H^P\circ S)_T$, $P$-a.s where $g$ is a representative of $j_P(X)$. The strategy depends on $P$ but not on the representative $g\in j_P(X)$. The equality $\Ccal = \bigcap_{P\in\mathcal{P}} j_P^{-1}\circ j_P(\Ccal)$ guarantees that in this case there exists a strategy $H\in \Hcal$ which is independent on $P\in\mathcal{P}$ such that $g\le (H\circ S)_T$, $P$-a.s for any $P\in \mathcal{P}$, where $g$ is any representative of $X$. The use of the universal filtration is crucial and guarantees the right measurability framework for the proof of these results.
\end{remark}

\begin{remark}[Pointwise framework] Theorem \ref{example:sensitive} can be obtained in the pointiwise setup proposed by \cite{Pointwise} using the superhedging duality. Indeed once the superhedging duality is obtained we can automatically deduce that sensitivity of the cone $\Ccal$. Also in this case one needs to extend the natural filtration in an opportune way in order to obtain an aggregation result for superhedging strategies.
\end{remark}

\paragraph{Quasi-sure aggregation in continuous time.}

The second case is an example of non-dominated volatility uncertainty (see e.g. \cite{Denis2,STZ11a,BD18}) which we briefly outline. For the sake of exposition, we restrict our attention to \cite[Example 4.5]{STZ11a}.
We set $C([0,T])$ the space of continuous functions on $[0,T]$ taking values in $\R$. Let $P^0$ be the Wiener measure on $\Omega=\{\omega\in C([0,T])\mid \omega(0)=0\}$.  Let $B:=\{B_t\}_{t\in[0,T]}$ be the canonical process, i.e. $B_t(\omega)=\omega_t$, $0\leq t\leq T$. The process $B$ is a standard Brownian motion under $P^0$ with respect to the rough filtration $\mathbb{F}=\{\cF_t\}_{0\leq t\leq T}:=\{\sigma(B_s\mid 0\leq s\leq t)\}_{0\leq t\leq T}$ and $\mathbb F^+=(\cF_t^+)_{0\leq t\leq T}$ its right continuous version. Recall that from \cite{Ka95} the quadratic variation can be defined pathwise and is given by the $\mathbb{F}$ adapted process $(\langle B\rangle_t)_{t\in[0,T]}$.
Following \cite[Example 4.5]{STZ11a} we consider a class of piecewise constant diffusion coefficients $\Vcal$ defined by $a=\sum_{n=0}^{\infty}a_n1_{[\tau_n,\tau_{n+1})}$, where $\{\tau_n\}_{n\in\N}$ is any non-decreasing sequence of $\mathbb{F}$ stopping times, with $\tau_0=0$, $\tau_n\leq T$ and $a_n$ being a positive valued  $\cF_{\tau_n}$ measurable random variable.  
Let $\cP:=\{\mathcal{P}^a\}_{a\in\Vcal}$ be the family composed by the measures $P^{a}=P^0\circ (X^{a})^{-1}$, where $X^{a}$ is the unique strong solution of \[                                                                                                                                                                                             
dX_t=a_t(X)dB_t\quad P^0\text{-a.s.}                                                                                                                                                                            \]
The existence of a strong solution for such a class is proved in \cite[Appendix]{STZ11a}. In particular, we have $\langle B\rangle_t= \int_{0}^t a_u^2 du$ $P^{a}$-a.s. for every $t\in [0,T]$ (see (4.10) in \cite{STZ11a}).  
For any probability $P$ we set $\mathcal{N}_t^{P}=\{A\subset B\mid B\in \cF_t \text{ and } P(B)=0\}$ and introduce the enlarged filtration
\begin{eqnarray}
	\mathbb{F}^{\Vcal} & \text{given by}& \cF^{\Vcal}_t  =  \bigcap_{a\in\Vcal} \cF_t^+\vee \mathcal{N}_t^{P^{a}}. \label{enlarged:t} 
\end{eqnarray}
Recall that that any $P^{a}$ uniquely extends to $\cF^{\Vcal}_t$  for any $t\in[0,T]$ and the filtration is still right continuous (see \cite{STZ11a}).

\begin{example} First we provide an example where $\cC\neq\tildeC{}$ unless we choose, in the discrete time model, the right continuous version of $\mathbb{F}$. Whereas in continuous time the use of the right continuous filtration is customary in discrete time is not.
	Consider a one period model by choosing two deterministic stopping times $0=\tau_0<\tau_1=1$.  Suppose that $\Vcal=[\underline{\sigma},\overline{\sigma}]$ for some $\underline{\sigma}<\overline{\sigma}$ non-negative, meaning that any plausible density of the quadratic variation process is constant and bounded in a given interval. 
	The class of corresponding probabilities is denoted by $\cP:=\{P^a\}_{a\in[\underline{\sigma},\overline{\sigma}]}$.
	Let $X:=B_1\mathbf{1}_{\{\langle B\rangle_1=\hat{a}^2\}}$ for some $\hat{a}\in[\underline{\sigma},\overline{\sigma}]$.
	We consider first the raw filtration $\mathbb{F}$ which implies that $\cF_0$ is trivial.
	We can easily see that $X\in\tildeC{}$. Indeed, for any $a\in[\underline{\sigma},\overline{\sigma}]$ with $a\neq\hat{a}$ we have $X=0$ $P^{a}$-a.s.\ with $0\in\cC$.
	Moreover, $X=B_1$ $P^{\hat{a}}$-a.s.\ with $B_1\in\cC$ as it corresponds to the buy and hold strategy of one unit of risky asset.
	We deduce
	\[X\in\bigcap_{P\in\cP} j^{-1}_P\circ j_P(\Ccal)=\tildeC{}.
	\]
	On the other hand it is not possible to find a trading strategy $H\in\R$ such that $H B_1\ge X$ $\qs$.
	Indeed $H B_1$ should be $P^a$ non-negative for any $\sigma\neq \hat{a}$, nevertheless, the $P^a$ distribution of $B_1$ is equal to the $P^0$ distribution of $a B_1$ (see \cite[Section 8]{STZ11a}).
	This implies $X\notin \cC$ and consequently $\cC\neq\tildeC{}$.
\end{example}

\begin{remark}
	It is worth pointing out that if one considers the $\cP$-completion of the right-continuous version of $\cF_t$, the sets $\{\langle B\rangle_1=a^2 \}\in\cF^{\Vcal}_0$ for every $a\in[\underline{\sigma},\overline{\sigma}]$.
	This implies $X\in\cC$, if $\cC$ is defined with respect to the filtration $\mathbb{F}^\Vcal$. 
\end{remark}

We consider again the set of processes $\mathcal V$ defined by $a=\sum_{n=0}^{\infty}a_n1_{[\tau_n,\tau_{n+1})}$, where $\{\tau_n\}_{n\in\N}$ is a non-decreasing sequence of $\mathbb{F}$ stopping times, with $\tau_0=0$, $\tau_n\leq T$ and $a_n$ being a positive  $\cF_{\tau_n}$ measurable random variable.  
Admissible strategies on the underlying process $B$ are given by the set of stochastic processes 
$$\Hcal(\mathbb{F}^{\Vcal})=\left\{H \text{ is } \mathbb{F}^{\Vcal} \text{-adapted and } \int_0^T|H_t|d\langle B\rangle_t<+\infty\right\}.$$
Here we are assuming the existence of a safe numeraire asset which pays $1$ at any time and we recall that the value of a portfolio is given by $V_t=h_t+H_tB_t$ 
where $h_t$ (respectively $H_t$) is the number of shares at time $t$ on the riskless (respectively risky) asset. 
\\For any $H\in \Hcal(\mathbb{F}^{\Vcal})$ the assumption of \cite[Theorem 6.4]{STZ11a} are satisfied and there exists an $\mathbb{F}^{\Vcal}$-adapted process $M$ such that $M_t=\int_0^t H_u dB_u$, $P^a$-almost surely for all $a\in \Vcal$ and therefore the following is well defined
\[\mathcal{K}(\mathbb{F}^{\Vcal}):=\left\{\int_0^T H_u dB_u\mid H\in \Hcal(\mathbb{F}^{\Vcal}), \int_0^t H_u dB_u\text{ is a } P^a\text{-supermtg } \forall a \in\Vcal\right\},\]
and under the standard self-financing condition we have $dV_t=H_t dB_t$.
With a slight abuse of notation we identify any $k\in \mathcal{K}(\mathbb{F}^{\Vcal})$ with the equivalence class $[k]\in \Lzero$ it generates. In this way we can consider  $\mathcal{K}(\mathbb{F}^{\Vcal})$ as a subset of $\Lzero$.

As customary, we define
\[\Ccal(\mathbb{F}^{\Vcal}):=\{X\in \Linf \mid X\leq  k   \text{ for some } k\in \mathcal{K}(\mathbb{F}^{\Vcal}) \}\]
and in the following we omit the dependence on the filtration.
\begin{proposition}\label{prop:sensitive} Consider the measurable space $(C[0,T], \cF_T^{\Vcal})$ and the class $\cP:=\{P^a\}_{a\in\Vcal}$. Then, the cone $\cC$ is sensitive, i.e.,
	\[ \Ccal = \bigcap_{P\in\cP} j_P^{-1}(j_P(\Ccal)).\]
\end{proposition}

\begin{proof} 
	As usual we denote by $\tildeC{}$ the right hand side.
	Since the inclusion $\subset$ is trivial we only need to show the opposite. Fix $X\in \tildeC{}$ and, for later use, let $\hat{X}:=X+\|X\|_{\infty}$ which is $\qs$ non-negative.
	Let $\mathcal{T}$ be the class of $\mathbb{F}^{\Vcal}$-stopping times. For $a,\nu\in\Vcal$ we set 
	\begin{center}$\theta^{a,\nu}:=\inf\left\{t\geq 0\mid \int_0^ta_u^2 du\neq \int_0^t\nu_u^2 du \right\}$\end{center}
	and, for $\tau\in \mathcal{T}$ and $a \in \Vcal$,
	\[\Vcal^{a}_{\tau}:=\{\nu\in\Vcal\mid \theta^{a,\nu}>\tau \text{ or } \theta^{a,\nu}=\tau=T\}.\]
	For $a\in\Vcal$ arbitrary, but fixed, we can thus consider the family $\{E_{P^{\nu}}[\hat{X}|\cF_{\tau}^{\Vcal}]\mid\nu\in \Vcal^{a}_{\tau}\}$ which is composed by $P^a$ essentially bounded elements. Set
	\begin{equation}\label{eq:esssup}
		V_{\tau}^{P^a}:= \sup_{\nu\in \Vcal^{a}_{\tau}}E_{P^{\nu}}[\hat{X}|\cF_{\tau}^{\Vcal}],
	\end{equation}
	where the supremum is computed as a $P^a$-essential supremum. For every $a \in \Vcal$ the family $\{V_{\tau}^{P^a}\}_{\tau\in\mathcal{T}}$ is uniformly integrable since $|V_{\tau}^{P^a}|\leq 2 \normc{X}$ $P^a$-a.s. for every $\tau\in\mathcal{T}$. 
	Therefore the assumptions of \cite[Theorem 7.1]{STZ11a} are met and we deduce,
	\begin{equation}\label{eq:duality} 
		\Pi(\hat{X})=\sup_{a\in\Vcal}\normc{V_0^{P^a}},
	\end{equation}
	where $\Pi(X)=\inf\left\{x\in\R\mid X-x\in \Ccal\right\}$ denotes the superhedging functional and where $V_0^{P^a}$ is as in \eqref{eq:esssup} with $\tau=0$.
	Recall now that the definition of $X\in\tildeC{}$ reads as 
	\[
	j_{P^a}(X)\leq j_{P^a}(k^a)\ P^a\text{-a.s. for some } k^a\in\mathcal K(\mathbb F^\Vcal), \text{ for any }a\in\Vcal.
	\]
	Thus, $X\in \tildeC{}$  guarantees that $X\leq \int_0^T H_t^{a}dB_t$ $P^a$-a.s.\ for every $a\in\Vcal$ which implies $\hat{X}\leq \normc{X}+\int_0^T H_t^{a}dB_t$ $P^a$-a.s.\ for every $a\in\Vcal$. The process $\int_0^t H_u^{a} dB_u$ is a 
	$P^a$-supermartingale. 
	Hence, we can conclude from $E_{P^a}[\int_0^T H_u^{a}dB_u|\cF_{0}^{\Vcal}]\leq 0$ for every $a\in\Vcal$ that
	\[0\leq V_0^{P^a}= \sup_{\nu\in \Vcal^{a}_0}E_{P^{\nu}}[\hat{X} \mid \cF_{0}^{\Vcal}] \leq \normc{X}+ \sup_{\nu\in \Vcal^{a}_0}E_{P^{\nu}}\left[\int_0^T H_u^{\nu}dB_u \mid \cF_{0}^{\Vcal}\right]\leq \normc{X}.\]
	By \eqref{eq:duality}, $\Pi(X+\normc{X})=\sup_{a\in\Vcal}\normc{V_0^{P^a}}\leq \normc{X}$, whence $\Pi(X)\leq 0$. 
	From \cite[Theorem 7.1]{STZ11a} and $\Pi(X+\normc{X})<\infty$ we know there exists a strategy $D\in \Hcal(\mathbb{F}^{\Vcal})$ such that 
	\[\Pi(X)+\normc{X}+\int_0^T D_tdB_t\geq \hat{X}\quad P^a\text{-a.s.},~a\in\Vcal,\]
	and $\int_0^t D_udB_u$ is a $P^a$-martingale for every $a\in\Vcal$. Finally
	\[\normc{X}+\int_0^T D_tdB_t\geq \rho_{\cC}(X)+\normc{X}+\int_0^T D_tdB_t\geq X+\normc{X}\quad P^a\text{-a.s.},~a\in\Vcal,\]
	and $X\leq \int_0^T D_tdB_t$, $P^a$-a.s.\ follows for every $a\in\Vcal$. 
\end{proof}

\section{Proofs of the main results}\label{sec:proofs}
Recall that $\normW{X}= \normc{X/W}$, for any $X\in\LW$ and $W\geq 1$.
Consider the set \[\cPW:=\{P\lll\cP\mid E_P[W]<+\infty\}.\] 
It is important to notice that $\cPW\approx \cP$. Indeed, for any $P\in \cP$, $W$ is integrable with respect to $\PW\sim P$ defined by $\frac{d\PW}{dP}= \frac{c}{W}$, where $c:=1/E_{P}[W^{-1}]$ is the normalizing constant.
\\For any $P\in\cPW$, we have $|E_P[X]|\leq E_P\left[\frac{|X|}{W}\cdot W\right]\leq \tilde{c} \normW{X}$ for $\tilde{c}=E_P[W]$, so that the linear functional $X\mapsto E_P[X]$ is continuous on $(\LW,\normW{\cdot})$ for any $P\in \cPW$. 
Let $\LPW\subset\LW^*$ be the span of the set of linear functional generated by $\cPW$ and $\LW^*$ be the topological dual of $\LW$.

We redefine the projection map $j_P$ in order to map $\LW$ on $\Linfp$ as
\[
\begin{array}{rccc}
\jp:&\LW& \to& \Linfp\\
& X&\mapsto&\left[\frac{X}{W}\right]_P
\end{array}
\]
The definition slightly differs from the one given in \eqref{projection:P}, but simple inspections show that this change does not affect the set $\tildeC{}$. In particular $$\tildeC{}=\bigcap_{P\in \cP} j_P^{-1}(j_P(\Ccal))=\bigcap_{P\in \cPW} j_P^{-1}(j_P(\Ccal))$$ 
and similarly for $\tildeC{\lambda}$. The map $j_P$  is easily shown to be continuous from $(\LW,\sigma(\LW,\LPW)$ to $(\Linfp,\sigma(\Linfp,L^1_P))$ if $P\lll \cP$.

\begin{lemma}\label{lem:cc} $\tildeC{}$, $\tildeC{\lambda}$ and $\widehat{\cC}=\cup_{\lambda\ge 0}\tildeC{\lambda}$ are monotone convex sets. In addition $\tildeC{}$ and $\widehat{\cC}$ are cones. 
\end{lemma}
\begin{proof}
	We only show the monotonicity of $\tildeC{}$, the other properties can be proven similarly.
	Suppose $Y\le X\ \qs$ with $X\in\tildeC{}$. By definition of $\tildeC{}$, for any $P\in\cP$, there exists $X^P\in\cC$ such that $X=X^P$ $P$- a.s. Take $Y^P=Y1_{\{X=X^P\}}+X^P1_{\{X\neq X^P\}}$ and observe that $Y^P\le X^P\ \qs$ From the monotonicity of $\cC$ we deduce $Y^P\in\cC$. Moreover, from $Y=Y^P$ $P$- a.s. and from $P\in\cP$ being arbitrary, $Y\in\tildeC{}$. 
\end{proof}

\subsection{Closure properties}\label{weak-closure}

In this subsection we consider   the sets $\cC,\CM$ defined by equations \eqref{cone:abstract} and \eqref{cone:abstract:M} respectively, and we will assume in Proposition \ref{closed:sensitive}, Lemma \ref{lem:p_closure} that $\cC$ is Fatou closed, which implies that $\CM$ is also Fatou closed.

\begin{proposition}\label{closed:sensitive}
	For every $P\ll \cP$ we have $\jp(\CM)$ is $\sigma(\Linfp,L^1_P)$-closed and therefore the set  $\bigcap_{P\in\cP'}
	j^{-1}_P(\jp(\CM))$ is  $\sigma(\LW,\lin(\cPW))$-closed for any $\cP'\subset \{P\lll \cP\}$. 
\end{proposition} 
The proof is based on the next two Lemmata.
For $\lambda,K>0$ define the set
\[
\CMK{K}:=\CM\cap \{X\in \LW \mid \normW{X} \leq 
K\},
\]

\begin{lemma}\label{lem:p_closure}
	For any probability $P\ll\cP$
	and for any $K\ge\lambda$ the set $\jp(\CMK{K})$ is
	$\sigma(\Linfp,\Lonep)$ - closed.
\end{lemma}

\begin{proof}
	Consider the continuous inclusion $$i: (\Linfp,\sigma(\Linfp,\Lonep))\to (\Lonep,\sigma(\Lonep,\Linfp)).$$ In a
	first step we show that $C(P):=i\circ \jp(\CMK{K})$ is
	closed in $\Lonep$ endowed with the usual norm $\|\cdot\|_{\Lonep}:=E_P[|\cdot|]$. To this end
	let $(Y_n)_{n\in \N}\subset C(P)$ and $Y\in \Lonep$ such that
	$\|Y_n-Y\|_{\Lonep}\to 0$, and without loss of generality we may also
	assume that $Y_n\to Y$ $P$-a.s. (by passing to a subsequence).
	Note that $|Y|$ is necessarily $P$-a.s. bounded by $KW$.
	Choose an arbitrary $X_n\in \CMK{K}$ such that $Y_n=\jp(X_n)$ for all $n\in \N$
	and an arbitrary $X\in \LW$ such that $Y=\jp(X)$. Consider the set 
	$$F=\{\omega\in \Omega\mid X_n(\omega)\rightarrow X(\omega)\}$$
	which satisfies $P(F)=1$. 
	Define $\widetilde{X}_n:= X_n 1_F -KW1_{F^c}\in \CMK{K}$ for $n\in\N$.
	By monotonicity of $\CM$, $\widetilde{X}_n$ for all
	$n\in \N$, and $ \widetilde{X}_n\rightarrow X 1_F
	-KW1_{F^c}=:\widetilde X$.
	Since $\CM$ is Fatou closed, the same holds for $\CMK{K}$. As a consequence, $\widetilde X
	\in \CMK{K}$. From $P(F)=1$ and the arbitrary choice of the representatives $X_n$ and $X$, we have
	$Y=\jp(X)=\jp(\widetilde X)\in C(P)$. Hence, $C(P):=i\circ
	\jp(\CMK{K})$ is $\|\cdot\|_{\Lonep}$-closed in $\Lonep$. As $C(P)$ is
	convex it then follows that $C(P)$ is
	$\sigma(\Lonep,\Linfp)$-closed and therefore $\jp(\CMK{K})$
	is $\sigma(\Linfp,\Lonep)$-closed by continuity of $i$.
\end{proof}

\begin{lemma}\label{lem:inclusion}
	For any probability $P\ll \cP$ we have the following representation
	$$\jp(\CM) = \bigcup_{K\ge\lambda} \jp(\CMK{K})$$
\end{lemma}

\begin{proof} Notice that by definition $\jp(\CM) \supset
	\jp(\CMK{K})$ for any $K\ge\lambda$ and hence $\jp(\CM) \supset
	\bigcup_{K\ge\lambda} \jp(\CMK{K})$. For the converse inclusion consider $Y\in \jp(\CM)$: there
	exists $X\in \CM$ such that $\jp(X)=Y$. Let $\bar
	K=\normW{X}$ then $Y\in \jp(\CMK{\bar K})$.
\end{proof}

\begin{proof}[Proof of Proposition \ref{closed:sensitive}]
	We first show that, for any $K\ge\lambda$,
	\begin{equation}\label{eq:KreinSmulian}
		\jp(\CM) \cap \{Y\in \Linfp\mid
		\|Y\|_{P,\infty}\leq K\} = \jp(\CMK{K})
	\end{equation}
	The inclusion $\supset$ is clear from Lemma \ref{lem:inclusion}.
	To show the equality, let
	$Y\in \jp(\CM)$ with $\|Y\|_{P,\infty}\leq K$. There
	exists $X\in \CM$ with $\jp(X)=Y$. Let $k\in\KM$ such
	that $X\leq k$ and notice that $(-KW) \vee X \wedge KW \leq k$  and
	$\jp((-KW) \vee X \wedge KW)=Y$.
	
	From Lemma \ref{lem:p_closure} the sets in \eqref{eq:KreinSmulian} are $\sigma(\Linfp, \Lonep)$-closed for every $K\ge\lambda$. The Krein-Smulian Theorem implies
	that $\jp(\CM)$ is $\sigma(\Linfp, \Lonep)$-closed and
	therefore $\jp^{-1}\circ \jp(\CM)$ is $\sigma(\LW,
	\lin(\cPW))$-closed.
	The last assertion follows by the intersection of closed
	sets.
\end{proof}

\subsection{Proof of the FTAP}
In this section we prove Theorem \ref{FTAP} and its discrete-time version Theorem \ref{FTAP_discrete}.
\begin{definition}\label{rhoA} For any set $\Acal\subset \Lzero$ and $X\in \Lzero$, we define $$\rho_{\Acal}(X):=\inf\{x\in\R \mid X-x\in
	\Acal\}.$$ 
\end{definition}
Note that for a monotone set $\Acal$ with $0\in\Acal$, $\rho_{\Acal}(X)<\infty$ for any $X\in\Linf$.

\begin{lemma}\label{strict:positive} The following are equivalent:
	\begin{enumerate}[i)]
		
		\item $\sNFLVR$;
		
		\item $\rho_{\tildeC{}}(\xi)>0$ for any
		$\xi\in \Linfplus\setminus\{0\}$;
		
		\item $\rho_{\tildeC{}}(1_A)>0$ for any $A\in\Fcal\setminus \polar$.
	\end{enumerate}
\end{lemma}

\begin{proof}$i)\Leftrightarrow ii)$: from Lemma \ref{lem:cc}, $\tildeC{}$ is monotone and $\rho_{\tildeC{}}=\rho_{\tildeC{}\cap \Linf}$ on $\Linf$. Thus, $\rho_{\cl(\tildeC{})}=\rho_{\tildeC{}\cap \Linf}=\rho_{\tildeC{}}$ on $\Linf$. The rest follows from $\cl(\tildeC{})=\{\rho_{\cl(\tildeC{})}\le 0\}=\{\rho_{\tildeC{}}\le 0\}$.

	$ii) \Rightarrow iii)$: it follows from $1_A\in \Linfplus\setminus\{0\}$.

	$iii) \Rightarrow ii)$: for any $\xi\in \Linfplus\setminus\{0\}$ we can find $n\in\N$
	such that $A:=\{\xi>1/n\}\in\Fcal\setminus \polar$.
	From Lemma \ref{lem:cc},
	$\rho_{\tildeC{}}$ is monotone and positive homogeneous. We deduce,
	\[
	0<n^{-1}\rho_{\tildeC{}}(1_A)=\rho_{\tildeC{}}(n^{-1}1_{\{\xi>1/n\}})\leq
	\rho_{\tildeC{}}(\xi).
	\]
\end{proof}

\begin{lemma}\label{approximating:martingale} 
	Consider now the conditions:
	\begin{enumerate}[a)]
		
		\item $\rho_{\tildeC{}}(1_A)>0$ for any $A\in\Fcal\setminus \polar$;
		
		\item for $A\in\Fcal\setminus \polar$ we can find $\delta\in
		(0,1]$ and 
		$\cQ=\{\QN\}_{n\in\N}\subset \cPW$, such that $\cQ\ll\bar{P}$ for some $\bar{P}\in\cP$ and
		
		$$\QN(A)\geq \delta\quad\text{ and }\quad E_{\QN}(X)\leq \frac{1}{n} \quad \forall\,X\in\cC_{n}, $$
		for any $n\in\N$.
	\end{enumerate}
	Then $b)\Rightarrow a)$. If, in addition, $\cC$ is Fatou closed $a)\Rightarrow b)$.
\end{lemma}

\begin{proof}
	b) $\Rightarrow$ a). 
	Suppose that there exists $A\in\Fcal\setminus \polar$ such that $\rho_{\tildeC{}}(1_A)\le 0$.
	Let $\delta\in (0,1]$, $\cQ$ and $\bar{P}$ as in b).
	From $\rho_{\tildeC{}}(1_A)\le 0$, it follows $1_A-\delta/4\in\tildeC{}$. In particular, by definition of $\tildeC{}$, we have
	$1_A-\delta/4\in j_{\bar{P}}^{-1}(j_{\bar{P}}(\cC))$. Thus, there exists $X\in\cC$ such that $X=1_A-\delta/4$ $\bar{P}$-a.s. More precisely, since $\cC=\cup_{\lambda\ge 0}\CM$ and $\CM$ is an increasing collection of sets, there exists $\bar{\lambda}\ge 0$ such that $X\in\cC_{\lambda}$ for every $\lambda\ge\bar{\lambda}$. 
	Moreover, since $\{Q_n\}\ll \bar{P}$, $X=1_A-\delta/4$ $\QN$-a.s.\ for any $n\in\N$.
	Using b), we deduce
	\[
	E_{\QN}[1_A-\delta/4]=E_{\QN}[X]\le \frac{1}{n},\qquad n\ge \bar{\lambda}.
	\]
	For $n\ge \bar{\lambda}\vee 4\delta^{-1}$, we have $\QN(A)\le \delta/2$ which contradicts $\QN(A)\ge \delta$.

	a) $\Rightarrow$ b).   
	Let $A\in\Fcal\setminus \polar$
	and $0<\delta< \rho_{\tildeC{}}(1_A)$.
	From $\rho_{\tildeC{}}(1_A-\delta)>0$, we deduce that $1_A-\delta\notin \tildeC{}$.
	Therefore we can find $\bar{P}\in\cP$ such that $1_A-\delta\notin j_{\bar{P}}^{-1}\circ j_{\bar{P}}(\cC)$ and, in particular, $1_A-\delta\notin j_{\bar{P}}^{-1}\circ j_{\bar{P}}(\cC_{n})$ for any $n\in\N$.
	We note now that the same is true for $\alpha (1_A-\delta)$ with $\alpha\in (0,1)$.
	Indeed, 
	\[
	\alpha (1_A-\delta)\in j_{\bar{P}}^{-1}\circ j_{\bar{P}}(\cC_n)\Rightarrow 1_A-\delta \in j_{\bar{P}}^{-1}\circ j_{\bar{P}}(\cC_{\lceil n/{\alpha}\rceil})
	\]
	which would be a contradiction. 
	All these considerations hold true by an equivalent change of measure, thus, we may assume $\bar{P}\in\cPW$.
	Therefore for any $\alpha\in (0,1)$, and $n\in\N$ $\alpha
	(1_A-\delta)\notin j_{\bar{P}}^{-1}\circ j_{\bar{P}}(\cC_n)$, which is $\sigma(\LW,\LPW)$ closed by Proposition \ref{closed:sensitive}.
	\\Consider the $\sigma(\LW, \LPW)$-compact and convex set $\Acal_n=\{\alpha (1_A-\delta)\mid \alpha \in [1/n,1]\}$. From the previous observation $\Acal_n\cap
	j_{\bar{P}}^{-1}\circ j_{\bar{P}}(\cC_n)=\emptyset$.
	For any $n\in\N$, there exists $\mu_n\in \LPW$ such  that 
	\begin{equation}\label{eq:separator}
		\sup_{X\in j_{\bar{P}}^{-1}\circ j_{\bar{P}}(\cC_n)}\mu_n(X) <\frac{\mu_n(1_A)-\delta}{n}.
	\end{equation}
	We observe that $\mu_n$ is positive and $\mu_n(1_{\Omega})=1$: indeed suppose that there exists $\xi\in\LWplus$ such that $\mu_n(\xi)<0$. From $-\LWplus\subset \cC_n$, $-a\xi\in\cC_n$ for any $a>0$, from which $\lim_{a\to\infty}\mu_n(-a\xi)=\lim_{a\to\infty}-a\mu_n(\xi)=\infty$ contradicts \eqref{eq:separator}.
	Similarly, $a(1_\Omega-1)\in\cC_n$ for any $a\in\R$, which implies $\mu_n(1_\Omega)=1$.
	We deduce that $\mu_n$ is the linear functional induced by some $\QN\in\cPW$.
	Moreover, for any $n\in\N$, we have
	\begin{itemize}
		\item $\QN\ll\bar{P}$. Otherwise let $B\in \cF$ such that $\bar{P}(B)=0$ and $\QN(B)>0$. From $1_B=0$ $\bar{P}$-a.s. we have $a1_B\in j_{\bar{P}}^{-1}\circ j_{\bar{P}}(\cC_n)$  for every $a>0$ and $\sup_{a>0}E_{\QN}[a1_B]=+\infty$ contradicts  \eqref{eq:separator}.
		\item $\QN(A)\geq \delta$. It follows from $0\in\cC_n$, which implies that the supremum in \eqref{eq:separator} is non-negative.
		\item $\sup_{X\in
			\cC_{n}}E_{\QN}[X]\leq \sup_{X\in j_{\bar{P}}^{-1}\circ j_{\bar{P}}(\cC_n)}E_{\QN}[X] \le\frac{1}{n}$ follows again by \eqref{eq:separator}.
	\end{itemize}
	Since $\bar{P}$ is the same for every $n$, $\cQ=\{\QN\}_{n\in\N}\ll\bar{P}\in \cP$, which concludes the proof of b).
\end{proof}

\begin{proof}[proof of Theorem \ref{FTAP}]
	It follows from Lemma \ref{strict:positive} and Lemma \ref{approximating:martingale}.
\end{proof}

\begin{proof}[proof of Theorem \ref{FTAP_discrete}]
	Clearly $\sNFLVR$ implies $\NA$ and from Lemma \ref{closure} $\cC$ is Fatou closed  so that the conclusion of Lemma \ref{approximating:martingale} is an equivalence. The first statement follows as in the proof of Theorem \ref{FTAP}.
	
	If we now assume that $\tildeC{}=\cC$ we have that $\sNFLVR$ is equivalent to $\NFLVR$, which is further equivalent to $\NA$ since, from Lemma \ref{closure}, $\cC\cap\Linf$ is  $\normc{\cdot}$-closed.
	The implication ($\Rightarrow$) follows directly from the first part of the Theorem. For the converse implication, let $V_T(H,h):=(H\circ S)_T+h\cdot\Phi$.
	Suppose that, for some $(H,h)\in\Hcal$, $V_T(H,h)\ge 0\ \qs$ (and hence $(H,h)\in\HM$ for every $\lambda\ge 0$). If there exists $P\in\cP$ such that $P(\{V_T(H,h)>0\})>0$ then for some $a>0$ the set $A=\{V_T(H,h)\geq a\}\in\cF\setminus \polar$. By assumption, there exist $\delta>0$ and $\cQ\in\Qapp$ such that $\inf_{Q\in\cQ}Q(A)=\delta> 0$. Consider $k\in\N$ and define the strategy $(\hat{H}^k,\hat{h}^k):=\frac{k}{a}(H,h)$. Notice that since $0\le V_T(\hat{H}^k,\hat{h}^k)\wedge K\leq V_T(\hat{H}^k,\hat{h}^k)$ and  $\cC_n$ is monotone, $V_T(\hat{H}^k,\hat{h}^k)\wedge K$ belongs to $\cC_n$ for every $n\in\N$.
	From Lemma \ref{approximating:martingale} b), for an arbitrarily fixed $Q\in\cQ$,
	\[
	0\le \sup_{k\in\N} E_{Q}[V_T(\hat{H}^k,\hat{h}^k)\wedge K]\leq 1  
	\]
	for every $K>0$. By monotone convergence theorem $E_{Q}[V_T(\hat{H}^k,\hat{h}^k)]\leq 1$ for any $k\in\N$. Since $E_{Q}[1_AV_T(\hat{H}^k,\hat{h}^k)]=\frac{k}{a}E_{Q}\left[1_{\{V_T(H,h)\geq a\}}V_T(H,h)\right]\geq k \delta$ we have that $\sup_{k\in\N}E_{Q}[V_T(\hat{H}^k,\hat{h}^k)]=\infty$ which is a contradiction.
\end{proof}


\subsection{The FTAP for $\sNABR$}
When $\cC=\tildeC{}$, both sets are also equal to $\widehat{\cC}:= \bigcup_{\lambda\ge 0} \tildeC{\lambda}$.
In the general case we could define \emph{sensitive No Arbitrage with Bounded Risk} $\sNABR$ as $\widehat{\Ccal}\cap \Linfplus=\{0\}$ and the the following relations are satisfied:
\[
\sNFLVR  \Rightarrow \sNA \Rightarrow \sNABR \Rightarrow \NA 
\]

A similar characterization holds for this condition.

\begin{theorem}\label{FTAP1}
	Let $\cC$ defined in \eqref{cone:abstract} be Fatou closed. The following are equivalent:
	
	\begin{enumerate}
		
		\item $\sNABR$
		
		\item For every non polar set $A$ and $n\in\N$, there exists $0<\delta_{n}\le\frac{1}{n}$ and $\QN\in\cPW$ such that 
		\[
		\QN (A)>n\delta_{n}\quad\text{and}\quad E_{\QN}[X]<
		\delta_n\quad \forall\,X\in \tildeC{n},
		\]
		for every $n\in\N$.
	\end{enumerate}
	
\end{theorem}
From the equivalent formulation, it is clear that the two notions $\sNFLVR$ and $\sNABR$ are very close.
One difference is that in Theorem \ref{FTAP}, the lower bound for $\QN(A)$ is uniform for the collection $\{\QN\}_{n\in\N}$.
Moreover, in Theorem \ref{FTAP} it also holds $\cQ\ll\bar{P}$ for some $\bar{P}\in\cP$.

\begin{lemma}\label{strict:positive1} Under the same assumptions of Theorem \ref{FTAP1} the following are equivalent:
	\begin{enumerate}[a)]
		
		\item $\bigcup_{\lambda\ge 0}\tildeC{\lambda}\cap \Linfplus=\{0\}$;
		
		\item $\rho_{\tildeC{\lambda}}(\xi)>0$ for any
		$\xi\in \Linfplus\setminus\{0\}$ and 
		$\lambda\ge 0$ ;
		
		\item $\rho_{\tildeC{\lambda}}(1_A)>0$ for any $A\in\Fcal\setminus \polar$ and 
		$\lambda\ge 0$.
	\end{enumerate}
\end{lemma}

\begin{proof}
	For ease of notation, denote $\rho_{\lambda}=\rho_{\tildeC{\lambda}}$. 
	
	$a)\Leftrightarrow b)$: the proof follows from $\tildeC{\lambda}$ being $\sigma(\LW,\LPW)$ closed and monotone.

	$b) \Rightarrow c)$: it follows from $1_A\in \Linfplus\setminus\{0\}$.
	
	$c) \Rightarrow b)$: We first show that $\rho_{\lambda}(\frac{1}{n}1_A)> 0$ for any $n\in\N$ and $A\in\Fcal\setminus \polar$. \\ The inequality $\ge$ is clear by monotonicity.
	Suppose, by contradiction, there exists $\bar{n}$ such that $\rho_{\lambda}(\frac{1}{n}1_A)=0$ for every $n>\bar{n}$. 
	Since $\tildeC{\lambda}$ is $\sigma(\LW,\LPW)$ closed by Proposition \ref{closed:sensitive}, we infer that $\frac{1}{n}1_A\in \tildeC{\lambda}$ for every $n>\bar n$. 
	By definition of $\tildeC{\lambda}$ , for every $P\in \cP$ there exists $k^P\in \KM$ such that $\frac{1}{n}1_A\leq k^P$ $P$-a.s. which implies $1_A\in \tildeC{n\lambda}$ i.e. $\rho_{n\lambda}(1_A)=0$. This would contradict $c)$. 
	\\Now for any $\xi\in \Linfplus\setminus\{0\}$ we can find $n\in\N$ such that $A:=\{\xi>1/n\}\in\Fcal\setminus \polar$.
	From Lemma \ref{lem:cc}, $\rho_{\lambda}$ is monotone, from which,
	\[
	0<\rho_{\lambda}(n^{-1}1_A)\leq
	\rho_{\lambda}(\xi).
	\]
\end{proof}

\begin{proof}[Proof of Theorem \ref{FTAP1}]
	Suppose $\sNABR$ holds. By Lemma \ref{strict:positive1}, for any $\lambda\ge 0$, $\rho_{\tildeC{\lambda}}(1_A)>0$ and any $A\in\Fcal\setminus \polar$.
	Since $\tildeC{\lambda}$ is $\sigma(\LW, \LPW)$ closed by Proposition \ref{closed:sensitive}, $1_A\notin \tildeC{\lambda}$ for any $\lambda\ge 0$ and the same is true for $\alpha 1_A$ for $\alpha\in (0,1)$.
	Indeed, $\alpha 1_A\in \tildeC{\lambda}$ would imply $1_A\in \tildeC{\lambda/\alpha}$, a contradiction.
	Thus, for any $n\in\N$ the $\sigma(\LW,\LPW)$ closed and convex set $\tildeC{n}$ and the $\sigma(\LW, \LPW)$-compact and convex set $\Acal_n=\{\alpha 1_A\mid \alpha \in [1/n,1]\}$ are disjoint. 
	For any $n\in\N$, there exists $\QN\in \cPW$ and $\delta_n\in[0,1]$, such  that 
	\begin{equation}\label{eq:separator1}
		\sup_{X\in \tildeC{n}}E_{\QN}[X]<\delta_n <\frac{\QN(A)}{n}.
	\end{equation}
	From $0\in\cC_n$, $\delta_n>0$ and the thesis follows.

	For the converse implication suppose that there exists $A\in\Fcal\setminus \polar$ such that $\rho_{\tildeC{\lambda}}(1_A)\le 0$ for some $\lambda\ge 0$.
	Since  $\tildeC{\lambda}$ is $\sigma(\LW, \LPW)$ closed, $1_A\in \tildeC{\lambda}$.
	Let $n\in\N$ with $n\ge\lambda$.
	Take $\QN\in\cPW$ such that $\QN(A)>n\delta_n$ and $E_{\QN}[X]< \delta_n$ for any $X\in \tildeC{n}$.
	From $1_A\in \tildeC{n}$, we deduce
	\[
	n\delta_n< \QN(A)< \delta_n
	\]
	which yields the contradiction $n<1$.
\end{proof}


\subsection{Proof of Theorem \ref{FTAP_compact}}
We only need to show that  $\sNFLVR$ implies $\mtg\neq\emptyset$ and $\cP\ll\mtg$, the rest follows from Theorem \ref{FTAP_discrete}.
Assume $\sNFLVR$ and introduce the notation $S_{1:t}=(S_1,\ldots,S_t)$. 
Theorem \ref{FTAP} ensures that 
for $A\in\Fcal\setminus \polar$ and $n\in\N$, we can find 
a collection $\cQ=\{\QN\}_{n\in\N}$ of probability measures  such that
\begin{equation}\label{eq:sep_tight}
	\QN(A)\geq \delta\quad\text{ and }\quad E_{\QN}[X]\leq \frac{1}{n} \quad \forall\,X\in\cC_{n}.
\end{equation} 
Moreover, there exists $P\in\cP$ such that $\cQ\ll P$. By Assumption \ref{assumption}, $P$ has support on some compact set $K^P$, hence, the collection $\cQ$ is tight and, by Prokhorov Theorem, is relatively compact.
From Lemma \ref{approximating:martingale}, $\cQ\subset \cPW$ and therefore we can define $\{P_n\}_{n\in \N}$ by 
$\frac{dP_n}{dQ_n}=\frac{W}{E_{Q_n}[W]}$.
Since $P_n\sim \QN$ and $\QN(K^P)=1$ with $K^P$ independent of $n$, the collection $\{P_n\}_{n\in \N}$ is also tight (and hence relatively compact).
Moreover, for $X\in \cC_n$ we have $E_{\QN}[X]=E_{Q_n}[W]E_{P_n}\left[\frac{X}{W}\right]$
We deduce that there exists a convergent subsequence of $\{P_n\}_{n\in\N}$ whose limit is denoted by $\bar{P}\in\cM_1$. Finally we define the measure $Q$ by $\frac{dQ}{d\bar{P}}=\frac{c}{W}$ where $c:=1/E_{\bar{P}}[W^{-1}]$.
\\ Note that, since $\mathbb{F}$ is the natural filtration of $S$, any $\cF_t$-measurable random variable $H$ can be written as $h(S_1,\ldots,S_t)$, for some Borel-measurable function $h:\R^{d\times t}\to\R$.
Consider now the sets
\begin{eqnarray*}
	\cY&:=&\bigg\{f(S_{1:t}):\Omega\to\R\mid f\in\cC_b(\R^{d\times t})\bigg\},\\
	\cX&:=&\bigg\{h(S_{1:t}):\Omega\to\R\mid h\in\cB_b(\R^{d\times t}),\ E_Q[h(S_{1:t})(S^j_{t+1}-S^j_t)]=0\ \forall j\bigg\},
\end{eqnarray*}
where $\cB_b(\R^{d\times t})$ denotes the space of bounded measurable functions on $\R^{d\times t}$ and $\cC_b(\R^{d\times t})$ those which are, in addition, continuous.
We aim at using a Monotone Class Theorem to deduce that $\cX$ contains all bounded $\cF_t$-measurable function.
First note that $\cY$ is a multiplicative class, namely, if $Y_1,Y_2\in\cY$ then $Y_1Y_2\in\cY$. 
Next, we show that $\sigma(\cY)=\cF_t$.
To see this let $\Gamma_n$ be a sequence of compacts such that, $\R^{d\times t}=\cup_{n\in\N}\Gamma_n$
By Urysohn's Lemma, for any $n,m\in\N$ there exists $f_{n,m}\in\cC_b(\R^{d\times t})$ such that 
$f_{n,m}(x)=x$ on $\Gamma_n$ and $f_{n,m}(x)=0$ on the complementary of $\Gamma_n+B^\circ_{\frac{1}{m}}(0)$, where $B^\circ_{\frac{1}{m}}(0)$ denotes the open ball with center in $0$ and radius $1/m$. 
Take now an open set $O\subset \R^{d\times t}_+$ and note that $S_{1:t}^{-1}(O)=\cup_{n\in\N}S_{1:t}^{-1}(O\cap\Gamma_n)$.
By construction of $f_{n,m}$,
\[
S_{1:t}^{-1}(O\cap\Gamma_n)=\bigcap_{m\in\N} \{\omega\in\Omega\mid f_{n,m}(S_{1:t}(\omega))\in O\cap\Gamma_n\}\in\sigma(\cY)
\]
From $O$ and $n$ arbitrary we deduce $\cF_t\subset\sigma(\cY)$.
The opposite inclusion is trivial.

The next step is to show that $\cY\subset\cX$.
Let $f\in\cC_b(\R^{d\times t})$ and define $X^j:=f(S_{1:t})(S^j_{t+1}-S^j_t)$ for $j=1,\ldots, d$.
By the choice of $W$ and $f$ bounded, we deduce that $X^j/W\in\cC_b(\Omega)$ and there exists $\bar{n}\in\N$ such that $X^j\in \cC_{n}$ for any $n\ge\bar{n}$. From \eqref{eq:sep_tight} and $W\geq 1$ we have 
\[ \frac{1}{n}\ge E_{\QN}[X^j] = E_{\QN}[W]E_{P_n}\bigg[\frac{X^j}{W}\bigg] \text{ implies } E_{P_n}\bigg[\frac{X^j}{W}\bigg]\le \frac{1}{n}.
\]
Using the weak convergence of $P_n$ to $\bar{P}$, we deduce $E_{\bar{P}}[X^j/W]\le 0$. By repeating the same argument for $-X^j$, we obtain $E_{\bar{P}}[X^j/W]= 0$  and hence $E_Q[X^j]=0$.

We now note that $\cX$ is a vector space and $1_{\Omega}\in\cY\subset\cX$. Moreover, for an increasing sequence $\{H_n\}_{n\in\N}\subset\cX$ with $\lim_{n\to\infty}H_n=H$ bounded, we have that $H\in\cX$ by dominated convergence.
From \cite[Theorem I.8]{P05} and $t\in\cI$ arbitrary, we conclude that $S$ is a $(Q,\mathbb{F})$-martingale.
The fact that $Q$ is calibrated to the prices of the options $\Phi$ follows from $\phi_i/W\in\tildeC{n}$ for every $i=1,\ldots,m$, $n\in\N$ and a similar weak convergence argument.

Finally we show that $\cP\ll\mtg$.
Let $A\in\Fcal\setminus \polar$, there exists $P\in\cP$ such that $P(A)>0$.
By \cite[Theorem 12.5]{AB06}, there exists a closed set $F\subset A$ such that $P(F)>0$. 
From $W\geq 1$ we have
\[
E_{P_n}[1_F]\ge E_{P_n}\bigg[\frac{1_F}{W}\bigg]=E_{\QN}[W]E_{\QN}[1_F]\ge \delta
\]
Due to Portemanteau's Theorem \cite[Theorem 15.3]{AB06}, $\bar{P}(F)\ge\limsup P_n(F)\ge \varepsilon\delta$. Thus, again by \cite[Theorem 12.5]{AB06}, $\bar{P}(A)\ge \bar{P}(F)>0$ and since $Q$ is equivalent to $\bar{P}$, we also have $Q(A)>0$. This concludes the proof.


\bibliographystyle{apalike}
\bibliography{uncertainty_bib}


\end{document}